%% file: article.tex
\documentclass[journal]{IEEEtran}
\usepackage{mathtools}

\input{macros}

\usepackage{cite}

\ifCLASSINFOpdf
\else
\fi
\usepackage{amsmath}
\usepackage{algorithm}
\usepackage{algpseudocode}
\ifCLASSOPTIONcompsoc
 \usepackage[caption=false,font=normalsize,labelfont=sf,textfont=sf]{subfig}
\else
 \usepackage[caption=false,font=footnotesize]{subfig}
\fi
\usepackage{graphicx}
\usepackage{amsfonts}
\usepackage{hyperref}
\hypersetup{
    colorlinks=true,
    linkcolor=blue,
    filecolor=magenta,      
    urlcolor=cyan,
    pdftitle={},
    pdfpagemode=FullScreen,
    }

\usepackage{xcolor}

\usepackage[final]{changes}

\usepackage{amssymb}
\usepackage{amsthm}
\usepackage{multicol}

\newtheorem{prop}{Proposition}
\newtheorem{corollary}{Corollary}[prop]

\theoremstyle{definition}

\newtheorem*{remark}{Remark}

\allowdisplaybreaks

\begin{document}
%
\title{Graph Diffusion-Advection Operator for\\ Directed Graph Signal Processing}
%
%
%

\author{Chun~Hei~Michael~Chan,~\IEEEmembership{Student Member,~IEEE}, %
        Alexandre~Cionca,~\IEEEmembership{Student Member,~IEEE}, 
     \\
    Viktor~Škultéty,~\IEEEmembership{Member,~IEEE}, Dimitri~Van~De~Ville,~\IEEEmembership{Fellow,~IEEE}\vspace*{-5ex}

\thanks{CHM. Chan, A. Cionca, V. Škultéty and D. Van De Ville are with the Neuro-X Institute, Ecole polytechnique fédérale de Lausanne, and the Department of Radiology and Medical Informatics, University of Geneva, Switzerland. E-mail: chunheimichael.chan@epfl.ch, dimitri.vandeville@epfl.ch}%
\thanks{This work was supported by the Swiss National Science Foundation, Sinergia project ``Precision mapping of electrical brain network dynamics with application to epilepsy'', Grant number 209470.}%
\thanks{Manuscript received XXX; revised XXX.}}

%
%

\markboth{Submitted to IEEE Transactions on Signal and Information Processing over Networks}%
{}
%



\maketitle

\begin{abstract}
Graph signal processing (GSP) provides a framework for analyzing data on irregular domains, with applications in neuroscience, finance, chemistry, and social sciences. Classical GSP primarily models symmetric relationships using undirected graphs, yet many real-world systems exhibit asymmetric interactions, motivating extensions to directed graphs. Central to directed GSP is the graph shift operator, typically defined via the directed graph Laplacian. Building on the well-known link between the undirected graph Laplacian and the diffusion operator, we establish a correspondence between the directed graph Laplacian and the diffusion-advection operator. This perspective opens new avenues for addressing crucial points such as frequency ordering, smoothness definition, and the design of spectral and graph filters. Specifically, we introduce two new orderings of frequencies based on the modulus and argument of the eigenvalues,  naturally leading to new definitions of smoothness. Then we present two kernels reflecting diffusive and advective processes, namely the heat and transport kernels, respectively. Finally, we propose novel graph filters obtained by composing  diffusive and advective parts, which approximate ideal spectral filters accurately and characterize the evolution of graph signals in richer ways. All aforementioned developments are illustrated on both synthetic and real graphs, including an application to temperature graph and signals.
\end{abstract}


%
\IEEEpeerreviewmaketitle

\section{Introduction}
%
%
%
%

\IEEEPARstart{D}{irected} graph signal processing (GSP) aims to analyze signals defined on irregular domains with asymmetric relationships, naturally captured by directed graphs. However, unlike the extension of classical signal processing to undirected graphs \cite{shuman_emerging_2013, ortega_graph_2018, marques_signal_2020}, where the symmetric graph Laplacian admits a real-valued orthonormal eigenbasis, directed graph shift operators (GSO) are asymmetric and, consequently, their adjacency and Laplacian matrices generally lack an orthogonal set of eigenvectors and exhibit complex-valued eigenvalues \cite{sandryhaila_algebraic_2012}.

To address these limitations, several approaches have been proposed. One line of work relies on the Jordan normal form (JNF) of the GSO to bypass potential diagonalization issues \cite{sandryhaila_discrete_2013}. Other approaches numerically retrieve bases that promote orthogonality and real-valued representations \cite{sardellitti_graph_2017, shafipour_directed_2019}. Alternative constructions introduce Hermitian Laplacians \cite{brefeld_graph_2020, wei_hermitian_2024} or leverage singular value decompositions \cite{chen_graph_2023, cionca2025community}, polar decompositions \cite{kwak_frequency_2024}, or Schur decompositions \cite{xiao_joint_2023}.
Despite these advances, a unified framework for directed GSP remains elusive, particularly one that yields an interpretable graph Fourier transform. 

Methods that retain complex-valued spectral representations often lack clear interpretation of the complex eigendecomposition \cite{sandryhaila_big_2014, singh_graph_2016, xiao_joint_2023}, whereas approaches enforcing real-valued bases forgo the ability to manipulate phase \cite{sardellitti_graph_2017, shafipour_directed_2019, brefeld_graph_2020, chen_graph_2023}. This limitation curtails the expressive power inherent in complex-valued graph spectral analysis, which is key to operations such as the graph Hilbert transform \cite{venkitaraman_hilbert_2019,chan2025hilbert} and graph surrogates generation \cite{chan2026statistical}. A further consequence of complex spectra is the ambiguity in frequency ordering. In the undirected case, eigenvalues can be naturally ordered and associated with smoothness via the Dirichlet energy \cite{ortega_graph_2018}. In contrast, directed graphs lack a consensus notion of frequency, which complicates the design and analysis of filters; e.g., previous work has  proposed ordering eigenvalues by their modulus~\cite{singh_graph_2016} or by their real part \cite{sevi_harmonic_2023}, or has introduced new smoothness measures such as directed variation \cite{sardellitti_graph_2017,shafipour_directed_2019}. 

Moreover, existing methods for directed graphs often overlook the original motivation of the graph Laplacian, namely, its interpretation as a discrete approximation of the Laplace operator in Euclidean spaces and of the Laplace-Beltrami operator on manifolds \cite{ortega_graph_2018, gine2006empirical}. To the best of our knowledge, this connection has not been explicitly nor systematically exploited in the directed GSP literature. Related ideas have emerged in graph neural networks (GNNs), albeit from a different angle based on attention layers; i.e., \cite{eliasof2024feature, wu2025supercharging} interpret directed edges in the adjacency matrix as modeling advection and undirected connectivity as modeling diffusion. 

These limitations motivate the need for a framework that explains the role of complex eigenvalues and eigenvectors in directed GSP. We take inspiration from the diffusion-advection operator in physics, which provides insight into interpreting the eigendecomposition of a directed graph Laplacian as the sum of two components: one associated with diffusion and the other with advection. First, we show how these components can be obtained through the spectral decomposition. Then, we propose a novel perspective on frequency ordering, based on the  modulus and the argument of the eigenvalues, which enables the introduction of new smoothness measures. We further show how the diffusion and advection operators can be linked to heat and transport kernels, respectively.  Finally, we propose novel graph filter designs  that are more flexible and expressive than conventional Laplacian-based polynomials. We demonstrate feasibility and effectiveness of the proposed framework for temperature on a real-world sensor network.



\section{Notations}
Lower bold cases denote vectors and $\vc x[n]$ denote its $n$-th element. Upper case bold letters denote matrix with $\ma U[n,m]$ denoting the $(n,m)$-th entry of $\ma U$. The notations $(.)^*$, $(.)^T$, $(.)^H$, $(.)^{-1}$, $(.)^{-H}$ denote conjugate, transpose, hermitian, inverse, and inverse hermitian, respectively. We use $j$ to refer to the complex number such that $j^2=-1$. $\Re(.)$ and $\Im(.)$ indicate real and imaginary parts, respectively. The entry-wise exponential of a vector $\vc x$ will be denoted $e^{\vc x}\in\mathbb{C}^N$ s.t $e^{\vc x}[n] = e^{\vc x[n]}$ and the operator $\Tr(.)$ refers to the trace of a matrix.

\section{Diffusion-Advection on Graphs}
We consider a directed graph $\mathcal{G}=(\mathcal{V},\mathcal{E},\mathcal{W})$, where $\mathcal{V}$ denotes the set of $N$ nodes, $\mathcal{E}\subseteq\mathcal{V}\times\mathcal{V}$ the set of directed edges, and $\mathcal{W}:\mathcal{E}\rightarrow\mathbb{R}^{+}$ assigns a positive weight to each edge. The graph is represented by the weighted adjacency matrix $\mathbf{W}\in\mathbb{R}^{N\times N}$ with entries $\mathbf{W}[n,m]=\mathcal{W}(n,m)$ if $(n,m)\in\mathcal{E}$, and $\mathbf{W}[n,m]=0$ otherwise.

\subsection{From Directed Graph Laplacian to Graph Diffusion-Advection Operator}
\label{sec:graph-diffusion-advection-operator}

\subsubsection{Directed Graph Laplacian}
We start by considering the directed graph Laplacian as the graph shift operator (GSO) \cite{singh_graph_2016}. The directed graph Laplacian is defined as $\ma L = \ma D - \ma W$, where $\ma D[n,n]=\sum_{m=1}^N\ma W[n,m]$ is the diagonal in-degree matrix. 
We subsequently consider its eigendecomposition
\begin{equation}
  \ma L=\ma U \ma \Lambda \ma U^{-1},
\end{equation}
where the eigenvalues $\ma \Lambda[n,n]=\ma \Lambda[n]$ are either real-valued and associated to real-valued eigenvectors $\ma U=[\ma u_n]_{n=1,\ldots,N}$, or complex-valued and then arranged in pairs of complex conjugate eigenvalues and eigenvectors. 


So far in GSP, this operator has been only associated to a ``skewed'' version of diffusion~\cite{mateos2019connecting}. However, due to symmetry constraints, a Laplacian operator associated with diffusion can only be made anisotropic, but not directional. 

\subsubsection{Diffusion-Advection Process}
We borrow fundamental insights from physics, specifically, the diffusion-advection process of an incompressible flow (divergence-free velocity field), with diffusion tensor $\ma D_\text{diff}(\vc x)$ and velocity field $\vc v(\vc x)$. The temporal evolution of the quantity $f(\vc x,t)$ is then governed by the following partial-differential equation~\cite{Batchelor2000-vd}:  
\begin{equation}
    \frac{\partial f}{\partial t} = \underbrace{(\nabla  \cdot ( {\ma D}_\text{diff} \cdot \nabla) - \vc v\cdot \nabla )}_{\mathcal{T}}f, \label{eq:adv-diff-pde}
\end{equation}
where $f\in C^{2}(\Omega \times [0,T))$ with $\Omega\subset \mathbb{R}^N$.
The diffusion term $\nabla \cdot ({\ma D}_\text{diff}) \cdot \nabla f$ represents the spreading of the scalar field due to local exchange, while the advection term $-\vc v\cdot \nabla f$ represents the transport of the scalar field by a velocity field $\vc v$, Importantly, the eigenfunctions and eigenvalues of the diffusion-advection operator $\mathcal{T}$ provide insights into the dynamics of the system. 
\begin{prop} \label{prop:adv-diff-eig}
    For isotropic and homogeneous diffusion $\ma D_\text{diff}=D_\text{diff}\ma {I}$ and constant velocity field ${\vc v}_0$, the analytical solution of the eigensystem
\begin{equation}
    \mathcal{T} \phi_{\vc k} = \lambda_{\vc k} \phi_{\vc k},
\end{equation}
    is well known~\cite{Arfken2005-nc}. Specifically, the eigenvalues $\lambda_{\vc k}$ and eigenfunctions $\phi_{\vc k}$ are of the form 
\begin{equation}
    \phi_{\vc k}({\vc x}) = e^{j {\vc k}\cdot {\vc x}}, \ \ \lambda_{\vc k} = - D_\text{diff}\|{\vc k}\|^2 - j({\vc v}_0\cdot {\vc k}), \label{eq:adv-diff-eig}
\end{equation}
where $\vc k$ is the wavevector. 
\end{prop}
Although the analytical solution is unfeasible for general diffusion and advection, the essential structure of the eigensystem is maintained~\cite{Drazin_Reid_2004}. In particular, the physical interpretation of the spectrum is preserved. The real part of the eigenvalues governs effective diffusivity of each mode by decay behavior, while their imaginary part governs effective advection by inducing phase shifts.


\subsubsection{Graph Diffusion-Advection Operator}

Inspired by the continuous case of Eq.~\eqref{eq:adv-diff-pde}, we propose to interpret the directed Laplacian as a diffusion-advection operator acting on a time-dependent graph signal $\vc x_t$: 
\begin{equation} \label{eq:diff-adv-L}
 \frac{\partial \vc x_t}{\partial t}=-\alpha \ma L \vc x_t,
\end{equation}
where $\alpha$ is a constant tuning the strength of the operator and $\vc x_t$ is defined as a map $\vc x_t: \mathcal{V}\times [0,T)\rightarrow \mathbb{R}$ such that $\partial_t \vc x_t$ exists for all $n\in\mathcal{V}$ and $t\in[0,T)$. 

Given the spectral properties of the physics operator, we propose to use the real and imaginary parts of the eigenvalues to split $\ma L$ into
\begin{equation} \label{eq:graph-adv-diff}
    \ma L = \underbrace{\ma U \ma \Lambda_R\ma U^{-1}}_{\ma L^{\circ}} + \underbrace{j\ma U \ma \Lambda_I\ma U^{-1}}_{\ma L^{\uparrow}},
\end{equation}
where $\ma \Lambda_R = \Re(\ma \Lambda)$ and $\ma \Lambda_I = \Im(\ma \Lambda)$. The new operators $\ma L^{\circ}$ and $\ma L^{\uparrow}$ represent the graph counterparts of diffusion and advection, respectively, by analogy with Eq.~\eqref{eq:adv-diff-eig}. The observation that $\ma L$ is composed of these constituents represents the starting point for our framework. 


\subsubsection{Fundamental Properties of \texorpdfstring{$\ma L^{\circ}$}{Lcirc} and \texorpdfstring{$\ma L^{\uparrow}$}{Lup}}
We highlight the main properties of the newly introduced operators $\ma L^{\circ}$ and $\ma L^{\uparrow}$. The proofs are given in Appendix~\ref{app:prop}. 
\begin{prop} \label{prop:real}
    $\ma L^{\circ}$ and $\ma L^{\uparrow}$ are real-valued.
\end{prop}

\begin{prop} \label{prop:zero-row-sum}
    $\ma L^{\circ}$ and $\ma L^{\uparrow}$ are zero-row sum matrices: $$\ma L^{\circ}\vc 1 = \vc 0 \text{ and } \ma L^{\uparrow}\vc 1 = \vc 0.$$
\end{prop}

\noindent 
Owing to this property we can retrieve the adjacency matrix associated to $\ma L^{\circ}$ and  $\ma L^{\uparrow}$ by subtracting their diagonal in-degree matrix.

\begin{prop} \label{prop:eigenvalues}
    $\ma L^{\circ}$ is of rank $N-1$ and has eigenvalues that are real and non-negative, while $\ma L^{\uparrow}$ has at most rank $N-1$ and purely imaginary eigenvalues.
\end{prop}

\begin{prop} \label{prop:trace}
    The traces of $\ma L^{\circ}$ and $\ma L^{\uparrow}$ are both real-valued and are such that by $\Tr(\ma L^{\circ}) > \Tr(\ma L^{\uparrow})=0$.
\end{prop}

\begin{prop} \label{prop:maximal-poly-diff}
    Let $N_C$ be the number of pairs of conjugate eigenvectors, the degree of the minimal polynomial of $\ma L^{\circ}$ is $\text{deg}(m_{\ma L^{\circ}}(\lambda))=N-N_C$.
\end{prop}

\begin{corollary} \label{corr:maximal-poly-diff}
    There exists a set of coefficients $c_k\in\mathbb{C}$ such that the matrix polynomial 
    \begin{equation}
        m_{\ma L^{\circ}}(\ma L^{\circ})=(\ma L^{\circ})^{N-N_C} + \sum_{k=0}^{N-N_C-1} c_k(\ma L^{\circ})^{k}=0.
    \end{equation}
\end{corollary}

\begin{prop} \label{prop:maximal-poly-adv}
    Let $N_Z$ be the number of eigenvalues with zero imaginary part, then it holds that the degree of the minimal polynomial of $\ma L^{\uparrow}$ is $\text{deg}(m_{\ma L^{\uparrow}}(\lambda))=N-N_Z+1$.
\end{prop}

\begin{corollary} \label{corr:maximal-poly-adv}
    There exists a set of coefficients $c_k\in\mathbb{C}$ such that the matrix polynomial 
    \begin{equation}
        m_{\ma L^{\uparrow}}(\ma L^{\uparrow})=(\ma L^{\uparrow})^{N-N_Z+1} + \sum_{k=0}^{N-N_Z} c_k(\ma L^{\uparrow})^{k}=0.
    \end{equation}
\end{corollary}

\noindent
Both Corollaries~\ref{corr:maximal-poly-diff} and~\ref{corr:maximal-poly-adv} exploit the specificity of the new operators to provide a tighter bound than the Cayley-Hamilton Theorem, dictating the maximal degree for spectral filter design through polynomial expressions. Notably the $\text{deg}(m_{\ma L^{\uparrow}}(\lambda))$ and $\text{deg}(m_{\ma L^{\circ}}(\lambda))$ can be different.

\subsection{On the Properties of Normal Graph Operator}
We study the case of normal graph operator. $\ma L$ is said to be normal if it satisfies $\ma L\ma L^T=\ma L^T\ma L$. This setting leads to a number of important structural properties, which we detail below. The proofs are given in Appendix~\ref{app:prop-normal}. 
\begin{prop} \label{prop:normal-symmetry-asymmetry}
Let $\ma L$ be a normal matrix. Then, its components $\ma L^{\circ}$ and $\ma L^{\uparrow}$ are symmetric and skew-symmetric, respectively, so we have
\begin{equation}
(\ma L^{\circ})^T = \ma L^{\circ}, \quad (\ma L^{\uparrow})^T = -\ma L^{\uparrow}.
\end{equation}
Moreover, these components admit the explicit decomposition
\begin{equation}
\ma L^{\circ} = \frac{1}{2}\left(\ma L + \ma L^T\right), \quad
\ma L^{\uparrow} = \frac{1}{2}\left(\ma L - \ma L^T\right).
\end{equation}
\end{prop}

The symmetric component $\ma L^{\circ}$ corresponds to a standard Laplacian operator, while the skew-Hermitian component $\ma L^{\uparrow}$ acts as a differential operator. As shown in Proposition~\ref{prop:zero-row-sum}, both $\ma L^{\circ}$ and $\ma L^{\uparrow}$ have zero row sum, which is a necessary condition for Laplacian and differential operators. Consequently, the decomposition in~\eqref{eq:graph-adv-diff} naturally aligns with the splitting of $\ma L$ into a symmetric diffusion component and a skew-symmetric advection component.

On a graph, the divergence of an operator quantifies the net flow imbalance at each node. When edge weights are viewed as flows \cite{jiang2011statistical}, the divergence of an operator $\ma L$ can be defined as the difference between incoming and outgoing edge weights \cite{grady2010discrete}
\begin{equation}
    (\nabla \cdot \ma L)[n] = \sum_{(n,m)\in\mathcal{E}} \ma L[n,m] \;-\;
 \sum_{(m,n)\in\mathcal{E}} \ma L[m,n].
\end{equation}
\begin{prop} \label{prop:normal-divergence-free}
If $\ma L$ is normal, then $\ma L^{\uparrow}$ is divergence-free; i.e., $\nabla \cdot \ma L^{\uparrow} = \vc 0$.
\end{prop}
The normality of $\ma L$ is consistent with the underlying assumption of divergence-freeness, or equivalently, incompressible flow in Eq.~\eqref{eq:adv-diff-pde}. 



\begin{remark}
    In graph theoretical terms, $\ma L^\uparrow$ is balanced (Eulerian) \cite{diestel2005graph}, which ensures the absence of topological sinks and sources \cite{grady2010discrete}, a structural condition that implies the existence of a cycle cover which in turn guarantees diagonalizability of the graph diffusion-advection operator \cite{chan2025hilbert}. 
\end{remark}


\subsection{Graph Signal Filtering with \texorpdfstring{$\ma L^{\circ}$}{Lcirc} and \texorpdfstring{$\ma L^{\uparrow}$}{Lup}}
We now revisit the concept when the operators are applied to a real-valued graph signal $\vc x\in \mathbb{R}^N$. First of all, the GFT of $\vc x$ is defined as 
\begin{equation}
    \hat{\vc x} = \text{GFT}\{ \vc x\} = {\ma U}^{-1}{\vc x}.
\end{equation}
Since $\vc x$ is real-valued, the resulting coefficients exhibit conjugate symmetry; i.e., $\hat{\vc x}[k_1]=\hat{\vc x}[k_2]^\star$ \cite{singh_graph_2016}. The inverse GFT is given by
\begin{equation}
    \vc x = \text{IGFT}\{ \hat{\vc x}\} = {\ma U}\hat{\vc x}.
\end{equation}
This now allows us to introduce graph-signal filtering in both the vertex domain and the graph spectral domain. 

\subsubsection{Graph Filtering} \label{sec:graph-filtering}
Graph filters are linear operators acting on graph signals and can be represented as a polynomial of the graph shift operator $\ma L$. Formally, a graph filter built upon an operator such as $\ma L$, is defined through a kernel function $g:\mathbb{C}^{N\times N}\rightarrow \mathbb{C}^{N\times N}$, such that
\begin{equation}
    g(\ma L) = g_0 \ma I + g_1 \ma L + \ldots + g_{K} \ma L^{K},
\end{equation}
where $g_k\in\mathbb{C}$ are the filter coefficients and $K$ is the order of the polynomial filter. The output of the filter when applied to a graph signal $\vc x$ is given by
\begin{equation}
    \vc y = g(\ma L) \vc x = \sum_{k=0}^K g_{k} \ma L^{k} \vc x.
\end{equation}

\subsubsection{Spectral Filtering} \label{sec:spectral-filtering}
Equipped with the GFT, we can characterize graph filters in the frequency domain leading to the frequency response of the filter. Thus, a graph filter $g(\ma L)$ can be expressed in the graph spectral domain as
\begin{equation}
    g(\ma L) =  \ma U \left(\sum_{k=0}^K g_k \ma \Lambda^k\right) \ma U^{-1} = \ma U g({\ma \Lambda}) \ma U^{-1},
\end{equation}
where $g({\ma \Lambda})$ is a diagonal matrix whose diagonal elements $g({\ma \Lambda})[n]$ are the frequency response of the filter at frequency index $n$.
Equivalently, one can first define the spectral filter $g({\ma \Lambda})$ and then apply it the graph frequency domain to Fourier coefficients $\hat{\vc x}$. The filtered signal $\vc y$ is given by
\begin{equation}
    \vc y = \text{IGFT}\{ g(\ma \Lambda) \hat{\vc x}\},
\end{equation}
where $\hat{\vc x} = \text{GFT}\{\vc x\}$.

\subsubsection{Heat and Transport Kernels}
Given the interpretation of the directed Laplacian as a diffusion-advection operator from Eq.~\eqref{eq:diff-adv-L}, we can find the solution to this system at any time $t>0$ for initial conditions $\vc x_0$ as $\vc y_t = g^\diamond_t(\ma L) \vc x_0$, where we define the graph filter equivalent to the heat-transport kernel
\begin{equation}\label{eq:trans-heat-graph}
    g^{\diamond}_t(\ma L)= e^{-\alpha t {\ma L}},
\end{equation}
which maps into the spectral domain as 
\begin{equation}\label{eq:trans-heat-spectral}
    g^{\diamond}_t(\ma \Lambda)= e^{-\alpha t {\ma\Lambda}}.
\end{equation}
Note that $g^{\diamond}_t$ admit an equivalent polynomial representation via its series expansion.
Interestingly, using $\ma L=\ma L^\circ + \ma L^\uparrow$, we can further decompose Eq.~\eqref{eq:trans-heat-graph} as
\begin{equation}
    g^\diamond_t(\ma L) = \underbrace{e^{-\alpha t {\ma L}^\circ}}_{g^\diamond_t(\ma L^\circ)} \underbrace{e^{-\alpha t {\ma L}^\uparrow}}_{g^{\diamond}_t(\ma L^\uparrow)},
\end{equation}
which allows us to define the heat and transport kernels that constitute $g^\diamond(\ma L)$, but can now be controlled separately. In addition, the spectral expressions are readily obtained as 
\begin{equation}
    g^{\diamond}_t(\ma \Lambda_R) = e^{-\alpha t \ma \Lambda_R}, \quad 
    g^{\diamond}_t(j\ma \Lambda_I) = e^{-j\alpha t \ma \Lambda_I}.  
\end{equation}
Thus the heat kernel affects the modulus of the GFT coefficients as a low-pass filter with exponential decay, while the transport kernel alters their phase with unit gain.

\begin{remark}
When considering a cycle graph of $N$ nodes and its eigenvalue decomposition, we have
\begin{equation*}
\ma \Lambda_R[n]=1-\cos\!\Big(\frac{2\pi n}{N}\Big),
\quad 
\ma \Lambda_I[n]=\sin\!\Big(\frac{2\pi n}{N}\Big),
\end{equation*}
from which one can retrieve the expressions of the the heat and transport kernels as 
\begin{eqnarray*}
g^{\diamond}_t(\ma \Lambda_R)[n] &\!\!\!\! = & \!\!\!\! e^{-\alpha t \left(1-\cos \left(\frac{2\pi n}{N}\right)\right)}
= e^{-\alpha t \left(\frac{2\pi n}{N}\right)^2/2}, \\
g_t^{\diamond}(j\ma \Lambda_I)[n] &\!\!\!\! = & \!\!\!\! e^{-j\alpha t \sin \left(\frac{2\pi n}{N}\right)}
= e^{-j\alpha t \frac{2\pi n}{N}}.
\end{eqnarray*}
In Fig.~\ref{fig:temporal-cycle}, the effect of these kernels is illustrated in the continuous domain and on the cycle graph. 
\end{remark}


\begin{figure}[t]
    \centering
    \subfloat[Diffusion: attenuation and dilation \label{fig:temporal-diff}]{%
       \includegraphics[width=1\linewidth]{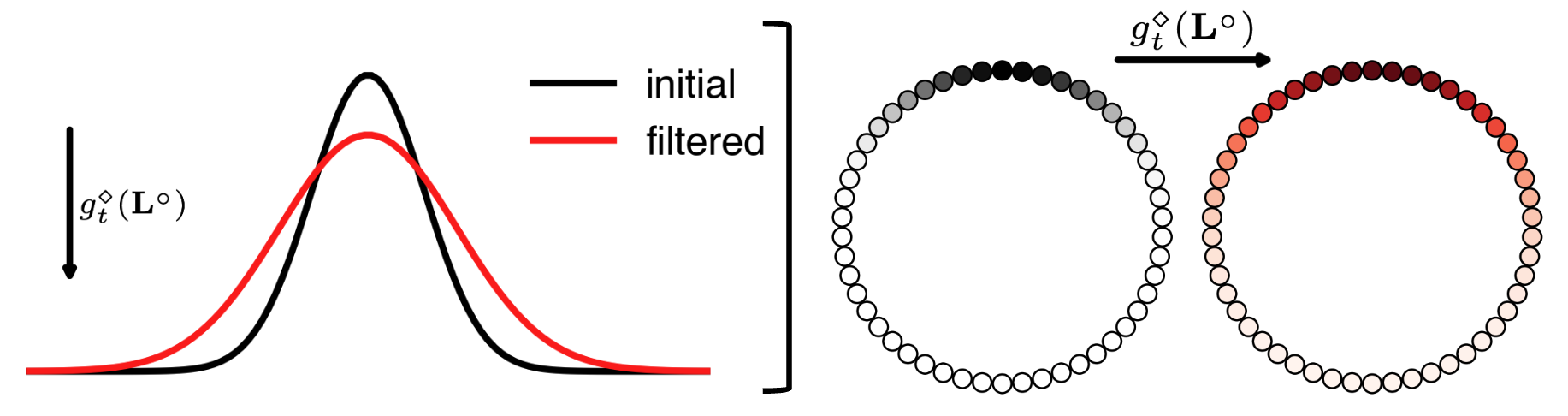}}
       \hfill
    \subfloat[Advection: translation \label{fig:temporal-adv}]{%
       \includegraphics[width=1\linewidth]{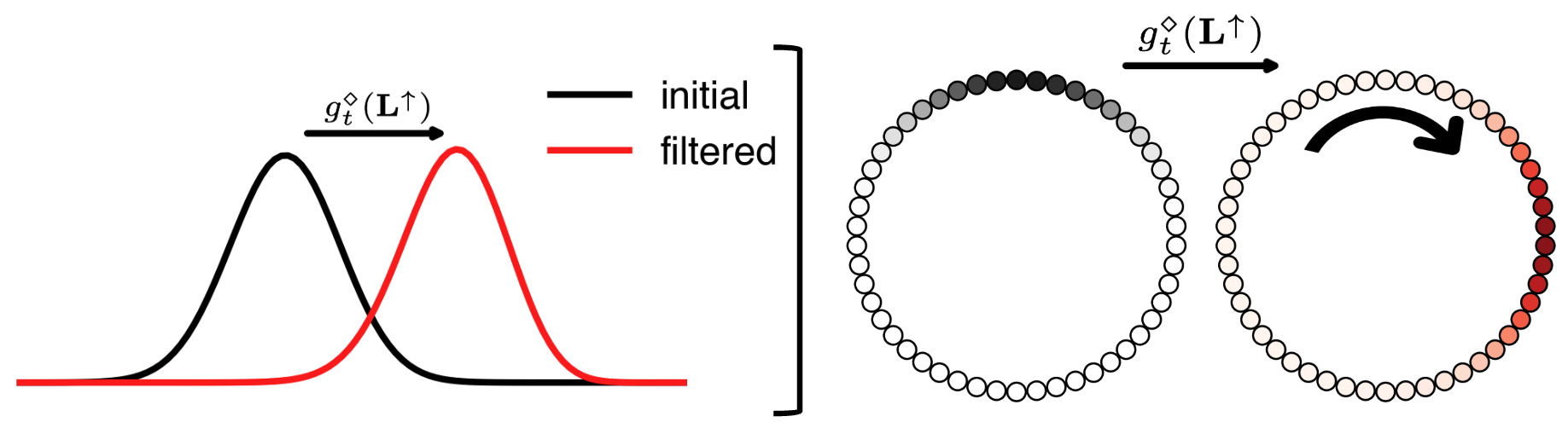}}\\
  \caption{\label{fig:temporal-cycle}Illustration of the action of the (a)~heat kernel and (b)~transport kernel on a signal in the continuous domain and on a cycle graph.}
\end{figure}


       


\section{Graph Signal Spectrum and Smoothness}
\subsection{Directed Graph Signal Spectrum}
While the GFT decomposes a signal into contributions by different eigenvectors that can be seen as a ``harmonic'' basis, the complex nature of the eigenvalues hinders their interpretation as ``frequencies'' and, in particular, how to order them. Given the analogy with the eigenvalue from the diffusion-advection operator in physics, we identify two natural ways of ordering. 

\subsubsection{Modulus-Based Ordering}
\label{sec:amplitude-ordering}
We first look at the dependency of the modulus of the eigenvalue with the norm of the wavevector, which can be seen as a frequency of the component across the graph; i.e., larger wavevectors for smaller wavelength, and vice versa. From Eq.~\eqref{eq:adv-diff-eig}, we can expand the modulus of $\lambda_{\vc k}$ as 
\begin{align} \label{eq:amplitude-ordering}
|\lambda_\vc k| &= \sqrt{\bigl(D_\text{diff} \|\mathbf{k}\|^{2}\bigr)^{2}
      + \bigl(\mathbf{v}\!\cdot\!\mathbf{k}\bigr)^{2}} \nonumber\\
&= \sqrt{D^{2}_\text{diff} \|\mathbf{k}\|^{4}
      + \|\mathbf{v}\|^{2}\,\|\mathbf{k}\|^{2}
        \cos^{2}\!\bigl(\angle(\mathbf{v},\mathbf{k})\bigr)} \nonumber\\
&\approx D_\text{diff}\,\|\mathbf{k}\|^{2},
\end{align}
where the last approximation holds due to the diffusion term that dominates; i.e., $D_\text{diff} \| \vc k \|^2 \gg |\vc v \cdot \vc k|$. As the modulus of the eigenvalue is proportional to the product of the diffusion coefficient and the wavevector $D_\text{diff}\|{\vc k}\|^2$, smaller values correspond to lower frequencies. Therefore, we propose radial ordering in the complex plane to rank frequencies according to increasing modulus of $\ma \Lambda[n]$.

\begin{remark}
    In undirected graphs, modulus ordering reverts to the conventional ordering according to increasing scalar eigenvalues.
\end{remark}

\subsubsection{Argument-Based Ordering}
\label{sec:angular-ordering}
Again taking inspiration from physics, we can derive from the argument of the eigenvalue that
\begin{eqnarray} \label{eq:angular-ordering}
  \!\!\!\!  \left|\tan(\angle \lambda_{\vc k}) \right| & = & \frac{\left|\vc v\cdot \vc k\right|}{D_\text{diff}\|\vc k\|^2} \nonumber\\
  & = & \frac{\|\vc v\| \left|\cos(\angle (\vc v,\vc k))\right|}{D_\text{diff}\|\vc k\|} \nonumber\\
  & \propto & \frac{\|{\vc v}\|}{D_\text{diff}\|\vc k\|}. 
\end{eqnarray}
Since the tangent of the angle is inversely proportional to $\|\vc k\|$, large $\left|\tan(\angle \lambda_\vc k)\right|$ are associated with a low oscillation frequencies. Furthermore, for the same $\left\|\vc k\right\|$, larger $\left|\tan(\angle \lambda_\vc k)\right|$ is associated to advection rather than diffusion, as shown by the factor $\|{\vc v}\|/D_\text{diff}$. 
Because the tangent is an increasing function, the ordering of $\tan(\angle \lambda_\vc k)$ is equivalent to that of $\angle \lambda_\vc k$. We therefore propose argument ordering to rank frequencies according to decreasing $
|\angle(\ma \Lambda[n])|$.

Given that the constant graph signal is an eigenvector of the GSO and is associated eigenvalue $0$, its angle is undefined, as it involves a division by zero. By convention, we instead set the constant graph signal as the first frequency in argument ordering. 
Therefore, the constant eigenvector is the lowest frequency eigenvector in both modulus and argument orderings. 

\begin{remark}
Modulus and argument orderings are equivalent in a cycle graph. 
\end{remark}


\subsection{Signal Smoothness} 
Ordering eigenvectors is essential to give meaning to the GFT, however, there is also the related concern to define smoothness of a graph signal. Thus, in analogy to the smoothness presented in previous works \cite{sandryhaila_discrete_2014, shafipour_directed_2019, singh_graph_2016, kwak_frequency_2024}, we define two measures to characterize the modulus- and argument-based smoothness by leveraging $\ma L^{\uparrow}$ and $\ma  L^{\circ}$, respectively:
\begin{eqnarray}
    \mathcal{S}^{\uparrow}(\vc x) = ||\ma L^{\uparrowcirc} \vc x||_2^2 \label{eq:angular-smoothness}, \quad 
    \mathcal{S}^{\circ}(\vc x) = ||\ma L\vc x||_2^2, \label{eq:radial-smoothness}
\end{eqnarray}
where $\ma L^{\uparrowcirc} = \ma L^{\uparrow} (\ma L^{\circ})^{\dagger}=\ma Uj\ma\Lambda_I\ma \Lambda_R^{\dagger}\ma U^{-1}$ and $(\cdot)^{\dagger}$ denotes the spectral pseudoinverse. These smoothness measures applied to the eigenvectors respect the proposed orderings as shown by the following Proposition. 

\begin{prop} \label{prop:smoothness-ordering} Modulus and argument orderings of eigenvalues induce equivalent smoothness relations of their corresponding eigenvectors:
\begin{align*}
|\ma \Lambda[n]| \leq |\ma \Lambda[m]|
& \;\Rightarrow\; 
\mathcal{S}^{\circ}(\vc u_n) \leq \mathcal{S}^{\circ}(\vc u_m), \\
|\angle(\ma \Lambda[n])| \leq |\angle(\ma \Lambda[m])|
& \;\Rightarrow\; 
\mathcal{S}^{\uparrow}(\vc u_n) \leq \mathcal{S}^{\uparrow}(\vc u_m),
\end{align*}
where $n,m\in [1,N]$. 
\end{prop}

Intuitively, for a graph signal $\vc x$, smaller values of $\mathcal{S}^{\circ}(\vc x)$ indicate smoothness over undirected parts of the graph, while larger values of $\mathcal{S}^{\uparrow}(\vc x)$ indicate a smooth signal following the directness.

\section{Filter Designs}
\label{sec:filters}

 

\subsection{Ideal Graph Filters}
\label{sec:filter-ideal}
The graph spectral domain can be conveniently used to specify and implement a graph filter. Specifically, for a graph signal $\vc x\in\mathbb{R}^N$ and a filter specified as a spectral window in the form of diagonal matrix $\hat{\ma H}\in\mathbb{C}^{N\times N}$, we can find the filtered signal $\vc y$ as
\begin{equation}
    {\vc y} = \text{IGFT}\left\{ \hat{\ma H}\, \text{GFT}\left\{{\vc x}\right\} \right\} = \ma U \hat{\ma H} \ma U^{-1} \vc x.
\end{equation}

We now introduce a number of ideal filters in the graph spectral domain. First, an ideal lowpass in the diffusive sense and based on the the modulus ordering of the eigenvalues is the diagonal matrix such that 
\begin{equation*}
    \hat{\ma H}_{\text{low}}^{\circ}[n, n]=
    \begin{cases}
    1 \text{,  if } n\in \mathcal{L}^\circ,\\
    0\text{,  otherwise,}
    \end{cases}
\end{equation*}
with $\mathcal{L}^\circ = \{n\in \mathbb{N}\ |\ |\ma \Lambda[n]| \leq T\}$ and where $T$ is a cutoff parameter. 
Similarly, a lowpass in the advective sense uses the argument-based ordering: 
\begin{equation*}
    \hat{\ma H}_{\text{low}}^{\uparrow}[n, n]=
    \begin{cases}
    1 \text{,  if } n\in \mathcal{L}^\uparrow,\\
    0\text{,  otherwise,}
    \end{cases}
\end{equation*}
with $\mathcal{L}^\uparrow = \{n\in \mathbb{N}\ |\ |\angle \ma \Lambda[n]| \geq T\}$. Bandpass or highpass filters can be easily defined using the same principle. Finally, we also introduce an ideal phase shift filter as
\begin{equation*}
    \hat{\ma H}_{\text{shift}}^{(q)}[n,n]= e^{j q \angle \ma \Lambda[n]},
\end{equation*}
where $q$ is the shift parameter. 

\subsection{Graph Shift Operator Filters}
\label{sec:gso-filter-approx}
To implement a graph filter ${\ma H}$ in vertex domain, given a spectral specification $\hat{\ma H}_\text{spec}$, one common approach is the  approximation by finite impulse response (FIR) graph filters \cite{shuman_chebyshev_2011, sakiyama_design_2017,segarra2015distributed}. In particular, we seek a $K$-th degree polynomial of $\ma L$:
\begin{equation*}
\tilde{h}(\ma L) = \sum_{k=0}^{K}\vct \gamma[k]\,\ma L^{k}.
\end{equation*}
Via the eigendecomposition of $\ma L=\ma U\ma \Lambda\ma U^{-1}$, the approximation of $\hat{ \ma H}_\text{spec}$ is formulated as the least-squares problem \cite{sakiyama_design_2017, segarra2015distributed}
\begin{equation} \label{eq:polynomial-fit}
\vct{\gamma}^{\star} =
\argmin_{\boldsymbol{\gamma}\in\mathbb{C}^{K+1}}
\Bigl\lVert
\hat{\ma H}_\text{spec} - \sum_{k=0}^{K}\vct \gamma[k]\,\ma \Lambda^{k}
\Bigr\rVert_2^{2}.
\end{equation}
By introducing the following Vandermonde matrix of order $K+1$,
\begin{equation*}
\ma V_K(\ma \Lambda)
=
\begin{bmatrix}
1 & \ma \Lambda[0] & \ma \Lambda^2[0] & \cdots & \ma \Lambda^K[0] \\
1 & \ma \Lambda[1] & \ma \Lambda^2[1] & \cdots & \ma \Lambda^K[1] \\
\vdots & \vdots & \vdots & \ddots & \vdots \\
1 & \ma \Lambda[N\!-\!1] & \ma \Lambda^2[N\!-\!1] & \cdots & \ma \Lambda^K[N\!-\!1]
\end{bmatrix},
\end{equation*}
we can rewrite Eq.~\eqref{eq:polynomial-fit} as
\begin{equation*}
\ma V_K(\ma \Lambda)\,\vct \gamma = \operatorname{diag}(\hat{\ma H}_\text{spec}),
\end{equation*}
and thus the least-squares solution is
\begin{equation}
\vct \gamma = \ma V_K(\ma \Lambda)^{\dagger}\,\operatorname{diag}(\hat{\ma H}_\text{spec}).
\end{equation}


       


\begin{figure*}[hbt!]
    \centering
    \subfloat[Ordering Comparisons \label{fig:ordering-frequencies}]{%
       \includegraphics[width=0.95\linewidth]{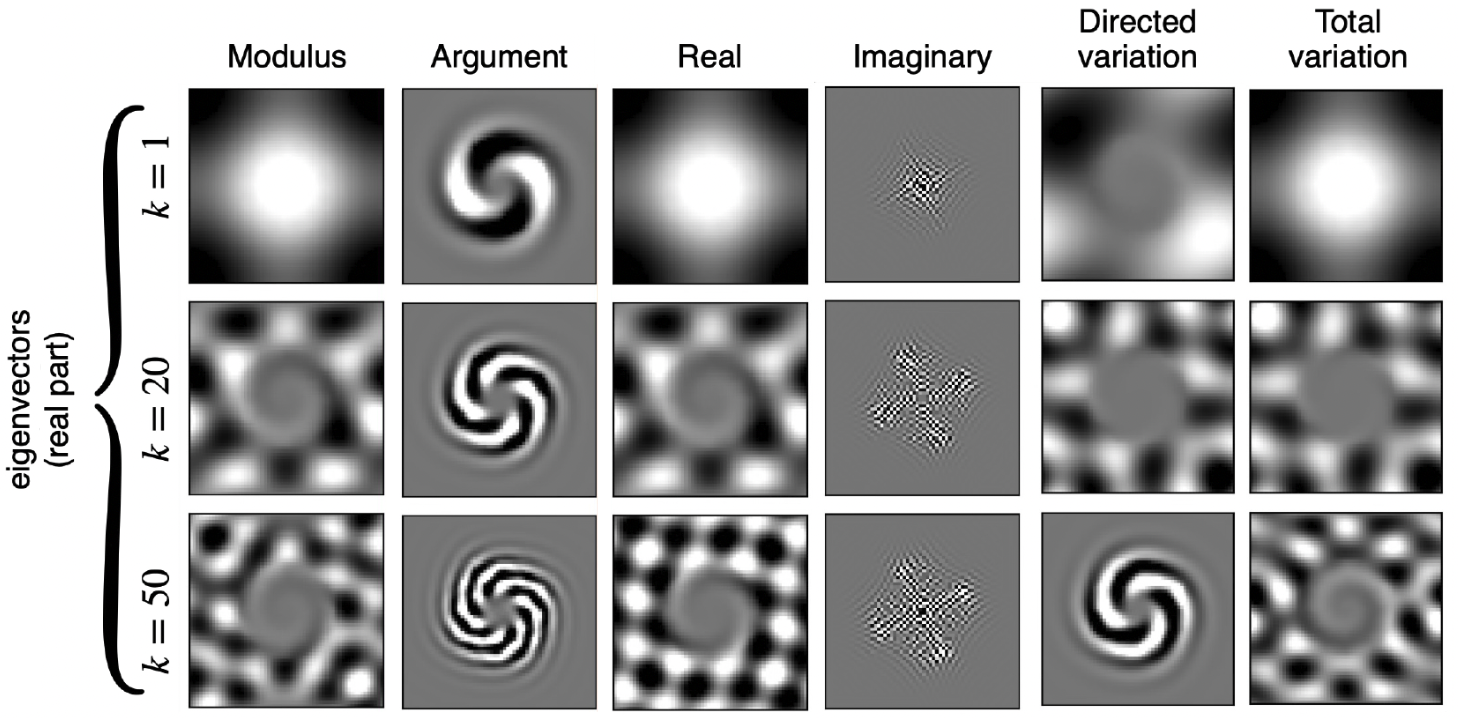}} \\
    \subfloat[Vortex Field \label{fig:flow-field-vortex}]{%
       \includegraphics[width=0.28\linewidth]{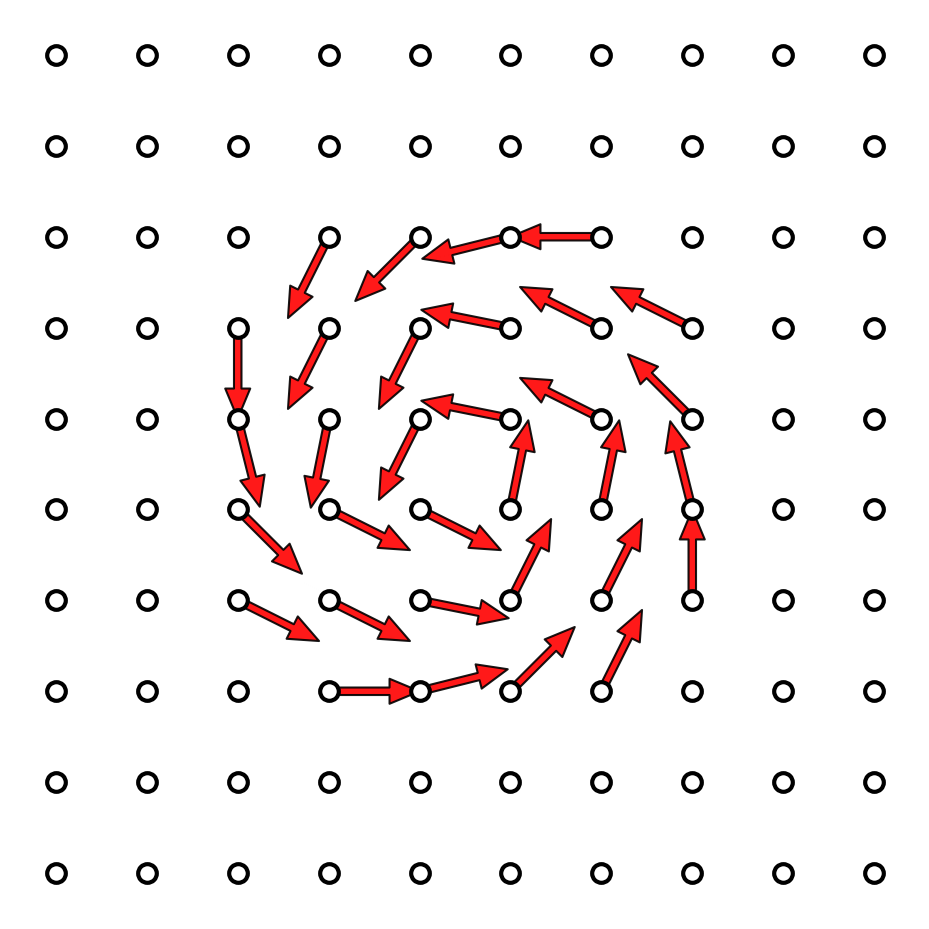}}
    \hspace{1em}
    \subfloat[Vortex Graph\label{fig:flow-field-vortex-graph}]{%
      \includegraphics[width=0.28\linewidth]{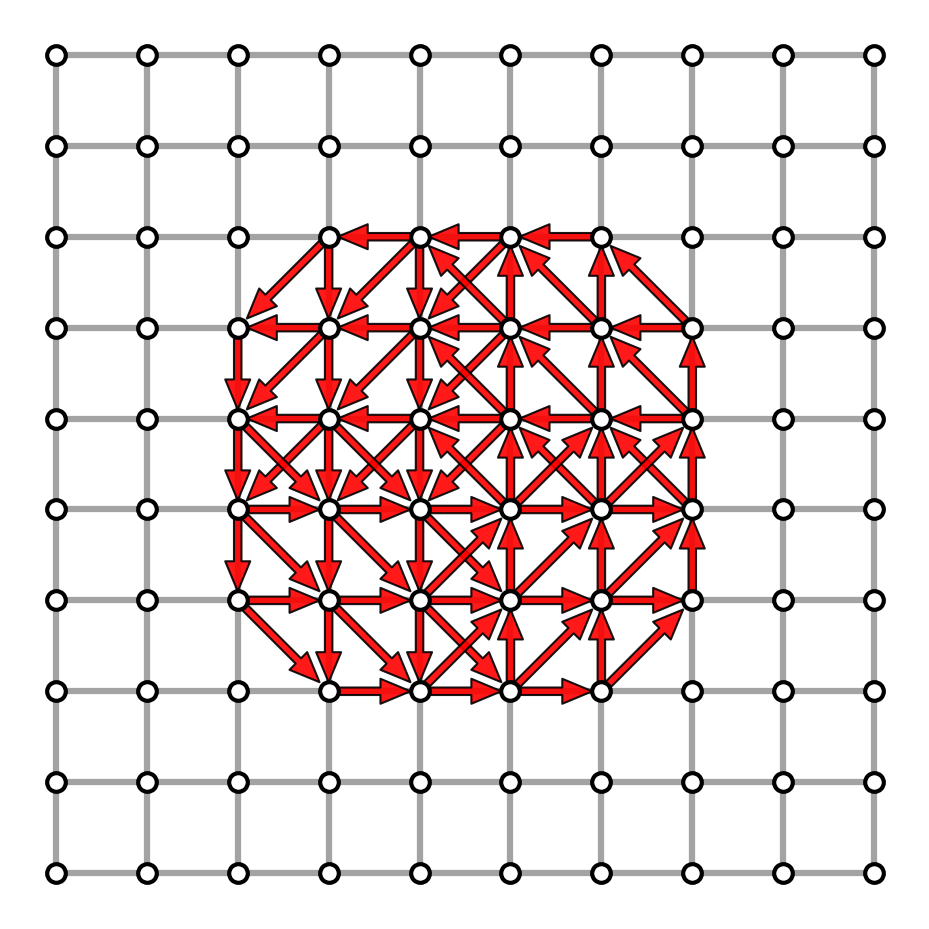}
    }
    \hspace{1em}
    \subfloat[Frequency Definitions \label{fig:eig-freqs}]{%
       \includegraphics[width=0.305\linewidth]{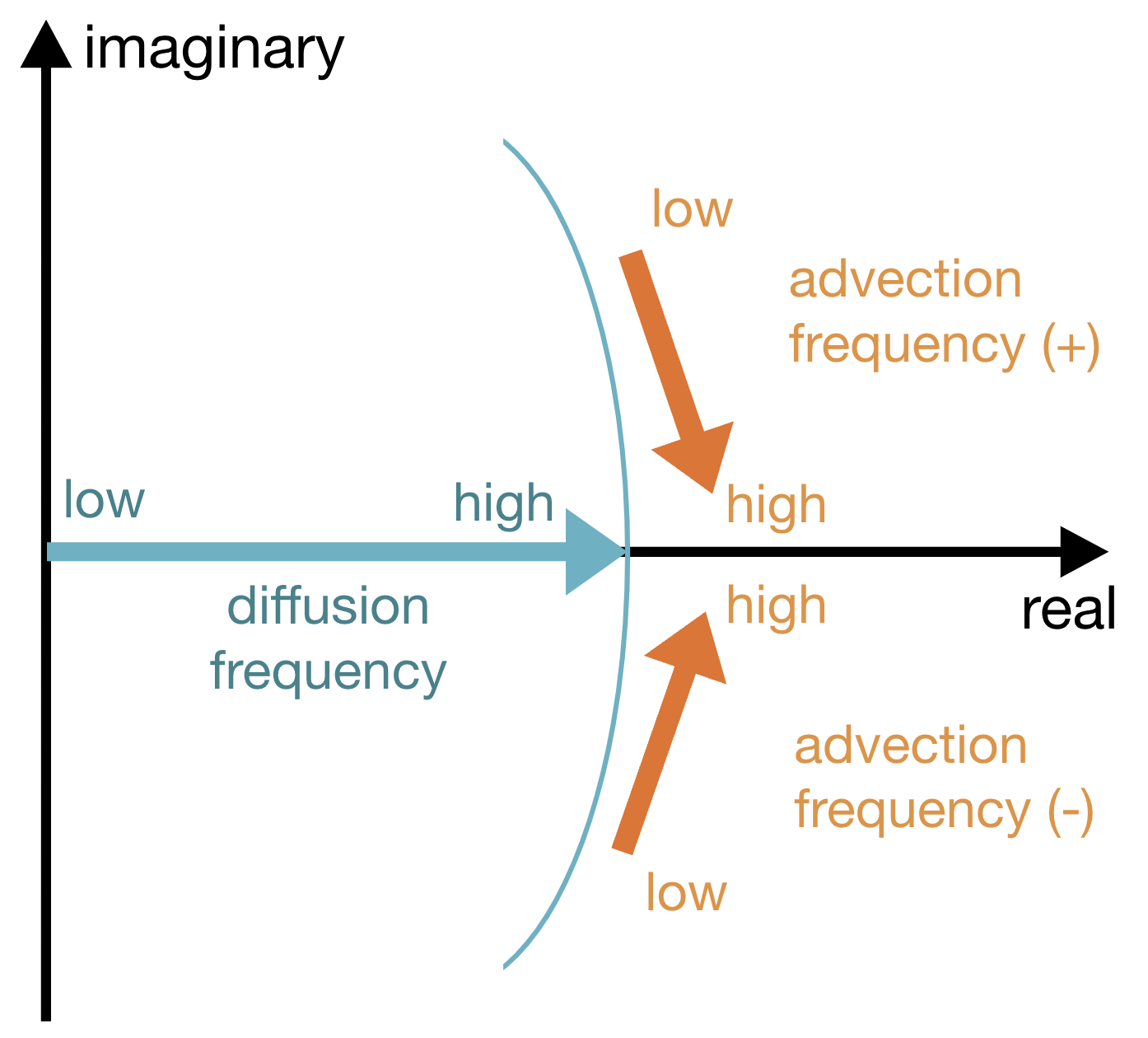}}
  \caption{(a) 1st, 20th, 50th, smoothest eigenvectors selected with due ordering computed from, in order from left to right, modulus, argument, real part, imaginary part of eigenvalues, directed variation (DV) and total variation (TV) of corresponding eigenvectors. (b) Vortex field showing the directionality per node. (c) Vortex graph with directed edges in red and undirected edges in gray. (d) Categorization of frequencies' smoothness, into modulus and argument smooth from corresponding eigenvalue.}
\end{figure*}

\subsection{Graph Sum Filters}
\label{sec:graph-sum-approx}
Given the two new operators $\ma L^{\uparrow}$ and $\ma L^{\circ}$, we propose a new form of $\ma H$ using separate polynomials for each operator. This form can then again be used to approximate a specified spectral window, but by the sum of a polynomial of $\ma \Lambda_R$ and a polynomial in $j\ma \Lambda_I$. Specifically, we define the sum-filter
\begin{equation}
    \tilde{h}_{\text{sum}}(\ma \Lambda_R, j\ma \Lambda_I)
    =
    \left(\sum_{k=0}^{K} \vct \alpha[k] \ma \Lambda_R^k \right)
    +
    \left(\sum_{k=0}^{K} \vct \beta[k] (j\ma \Lambda_I)^k \right),
\end{equation}
where $\ma \Lambda_R$ and $\ma \Lambda_I$ denote the real and imaginary parts of the spectrum, respectively. The corresponding vertex-domain filter is given by
\begin{equation}
    \tilde{h}_{\text{sum}}(\ma L^{\circ}, \ma L^{\uparrow})
    =
    \left(\sum_{k=0}^{K} \vct \alpha[k] (\ma L^{\circ})^k \right)
    +
    \left(\sum_{k=0}^{K} \vct \beta[k] (\ma L^{\uparrow})^k \right),
    \label{eq:graph-sum-filter}
\end{equation}
which explicitly separates the diffusive and advective contributions. To determine the coefficients $\vct \alpha$ and $\vct \beta$, we define
\begin{equation}
    \ma C_K(\ma \Lambda_R, j\ma \Lambda_I)
    =
    \big[\, \ma V_K(\ma \Lambda_R) \;\; \ma V_K(j\ma \Lambda_I) \big]
    \in \mathbb{C}^{N \times 2(K+1)},
\end{equation}
where $\ma V_K(\cdot)$ denotes the Vandermonde matrix of order $K+1$.
The coefficients are obtained via least-squares approximation of the desired spectral filter $\hat{\ma H}_\text{spec}$:
\begin{align}
    \vct \alpha[k]
    &= \left(\ma C_K(\ma \Lambda_R, j\ma \Lambda_I)^{\dagger} \, \operatorname{diag}(\hat{\ma H}_\text{spec})\right)[k], \\
    \vct \beta[k]
    &= \left(\ma C_K(\ma \Lambda_R, j\ma \Lambda_I)^{\dagger} \, \operatorname{diag}(\hat{\ma H}_\text{spec})\right)[k+K+1],
\end{align}
for $k\in[0,K]$. The degrees of the polynomials in $\ma L^{\circ}$ and $\ma L^{\uparrow}$ provide a quantitative measure of the locality of the heat and transport components, respectively. Higher degrees correspond to increasingly nonlocal behavior of the associated process.




\begin{figure*}
    \includegraphics[width=\linewidth]{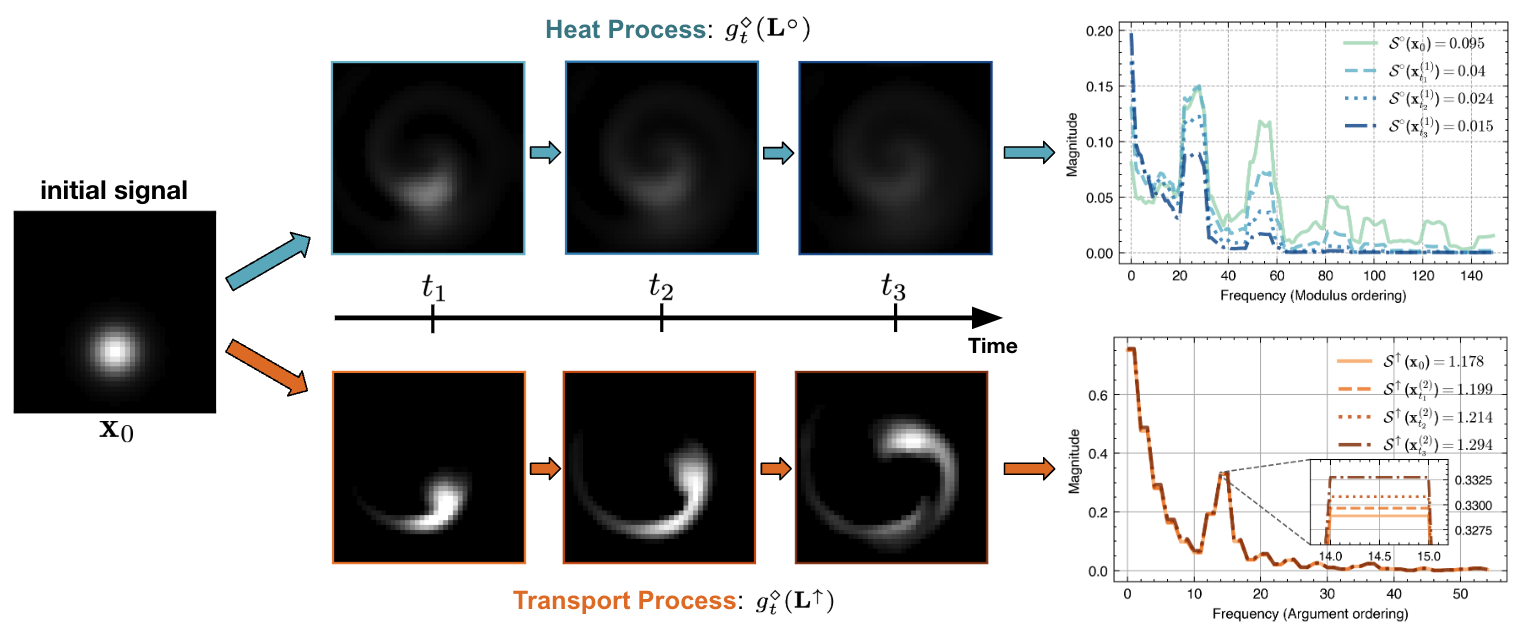}%
    \caption{Graph signals showing $\vc x_0$ (initial signal), diffused (upper row) and transported (lower row) signals at time $t_1, t_2$ and $t_3$. On the rightmost are the graph signals' spectra.}
    \label{fig:vortex-heat-transport-process}
\end{figure*}

\subsection{Graph Rational Filters}\label{sec:graph-rat-approx}
Given our new notion of smoothness based on the argument, we propose to use the associated operator $\ma L^{\uparrowcirc}$ to define a graph rational filter as a polynomial of this operator with the following expressions in the vertex and spectral domains: 
\begin{eqnarray}
    \tilde{h}_{\text{rat}}(\ma L^{\uparrowcirc}) & = &\sum_{k=0}^K\vct \eta[k] \left(\ma L^{\uparrowcirc}\right)^k, \label{eq:graph-rational-filter}\\
    \tilde{h}_{\text{rat}}(j\ma \Lambda_I \ma \Lambda_R^{\dagger}) & = &\sum_{k=0}^K\vct \eta[k] \left(j\ma \Lambda_I \ma \Lambda_R^{\dagger}\right)^k.
\end{eqnarray}
Again, as before when a desired spectral window is to be approximated, we can estimate $\vct \eta[k]$ by 
\begin{equation}
    \vct \eta[k] = \left(\ma V_K(j{\ma \Lambda_I \ma \Lambda_R^{\dagger}})^{\dagger}\, \operatorname{diag}(\hat{\ma H}_\text{spec})\right)[k],\quad k\in[0,K].
\end{equation}


\subsection{Filters Properties} \label{sec:LSI-filters}
\begin{prop} \label{prop:LSI-filters}
    The filters in~Eqs.~\eqref{eq:graph-sum-filter} and \eqref{eq:graph-rational-filter} are linear shift-invariant (LSI); i.e., they commute with $\ma L$. Additionally they commute with $\ma L^\circ$ and $\ma L^\uparrow$.
\end{prop}

\begin{prop} \label{prop:poly-real-coefs}
    Let $\hat{\ma H}\in \mathbb{C}^{N\times N}$ be a spectral filter. If, for every graph signal $\vc x\in\mathbb{R}^N$, the filtered signal~\upshape IGFT$\{\hat{\ma H}\hat{\vc x}\}\in\mathbb{R}^N$ then the coefficients $\vct \alpha, \vct \beta, \vct \eta\in\mathbb{R}^{K+1}$.
\end{prop}

For real-valued graph signals, this property restricts  the admissible filter coefficients to real values.


\section{Experimental Results}

\subsection{Vortex Graph}

We first illustrate the proposed frequency ordering using a vortex flow sampled on a two-dimensional grid, as shown in  Fig.~\ref{fig:flow-field-vortex}. Specifically, we consider a graph composed of $N$ nodes uniformly distributed on a square grid, where each node is characterized by its radial distance $r_n$ and angle $\theta_n$ from the grid center. The underlying support is constructed by connecting each node to its four nearest neighbors through undirected edges of unit weight, yielding a regular grid structure. To encode a vortex flow pattern, we introduce directed edges between node pairs whose radial distance satisfy $|r_n-r_m|<d$, where $d$ is one grid spacing. This is to account for the square-grid sampling, which does not permit exact circular alignment of nodes around the center. The orientation of these edges follows a counter-clockwise direction with respect to the grid center. In practice, an edge from node $n$ to node $m$ is assigned when the angular difference $\theta_m-\theta_n$ is negative. All added edge weights are set to 1. An illustrative example with $N=100$ nodes is shown in Fig.~\ref{fig:flow-field-vortex-graph}. In the subsequent analysis, we consider an upsampled version with $N=50\times 50=2500$ nodes.

\subsection{Frequency Orderings}
We compare the proposed modulus and argument orderings against other proposals from the literature, in particular, directed variation (DV)~\cite{sardellitti_graph_2017, shafipour_directed_2019} and total variation (TV)~\cite{sandryhaila_discrete_2014}:
\begin{eqnarray}
    \text{DV}(\vc x) & = & \sum_{n,m=1}^N \ma A[m,n]\, \max(\vc x[n] -\vc x[m], 0), \\
    \text{TV}(\vc x) & = & \|\vc x-\ma A\vc x\|_1.
\end{eqnarray}
In Fig.~\ref{fig:ordering-frequencies}, we illustrate the effect on the previously introduced vortex graph when ordering eigenvectors of the directed Laplacian according to increasing values of these smoothness measures. Complemented by the ordering based on real and imaginary part, modulus, and argument of the eigenvalues, we observe distinct behavior. Eigenvectors that are smooth according to modulus ordering exhibit slow oscillations in the undirected graph parts where they display approximately isotropic behavior. In contrast, eigenvectors that are smooth according to argument ordering show slow oscillatory patterns aligned with the flow field. Other orderings tend to mix these two features that we relate to the notions of diffusion and advection, respectively, essentially obscuring the contributions of the flow-aligned patterns that emerge from argument ordering.

\begin{figure*}[hbt!]
    \centering
    \captionsetup[subfigure]{justification=centering}
    \subfloat[Wind Field\label{fig:wind-field-temperature}]{%
       \includegraphics[width=0.25\linewidth]{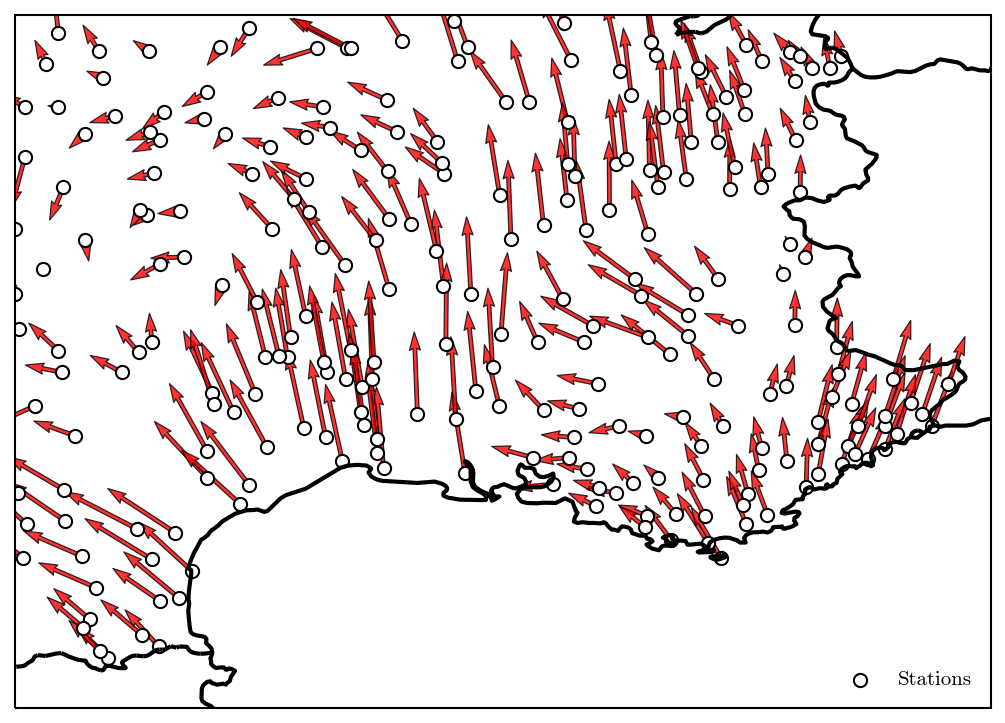}}
    \subfloat[Temperature Graph \label{fig:distance-based-temperature}]{%
       \includegraphics[width=0.25\linewidth]{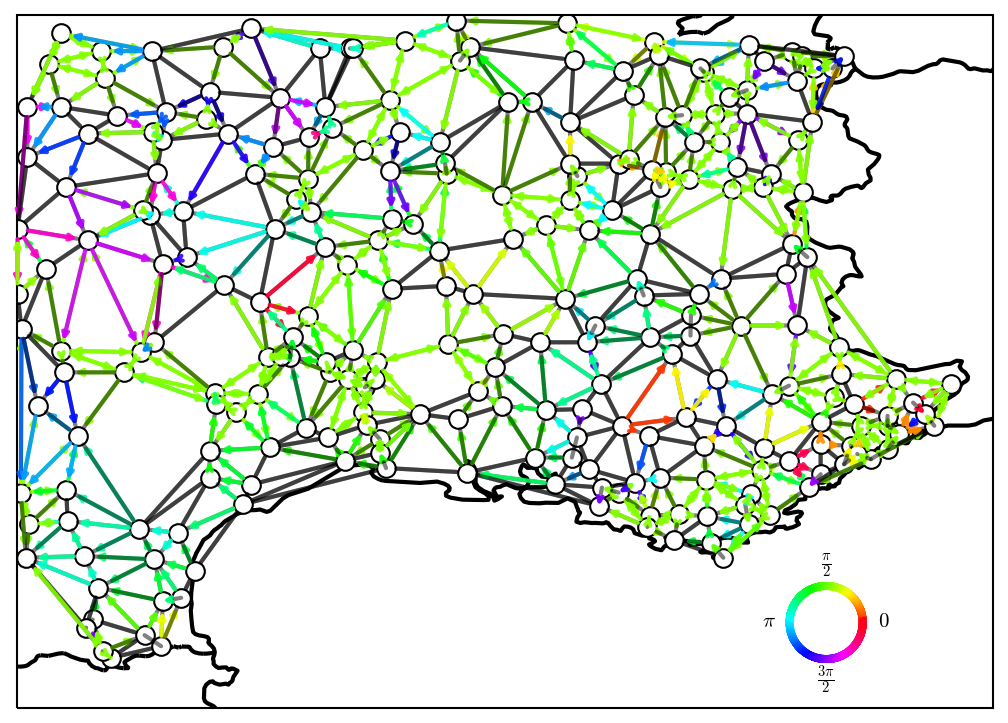}}
       \hfill
    \subfloat[Lowpass Diffusion \label{fig:graph-filterlowpass-D}]{%
       \includegraphics[width=0.5\linewidth]{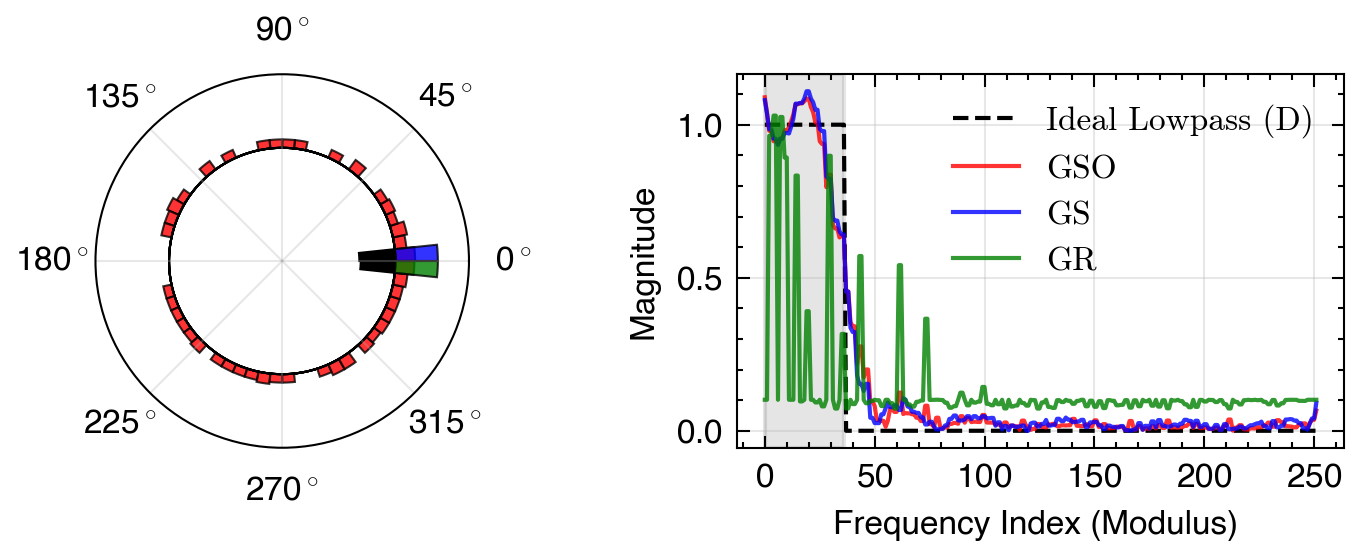}}
       \hfill\\
    \subfloat[Lowpass Advection \label{fig:graph-filterlowpass-A}]{
       \includegraphics[width=0.5\linewidth]{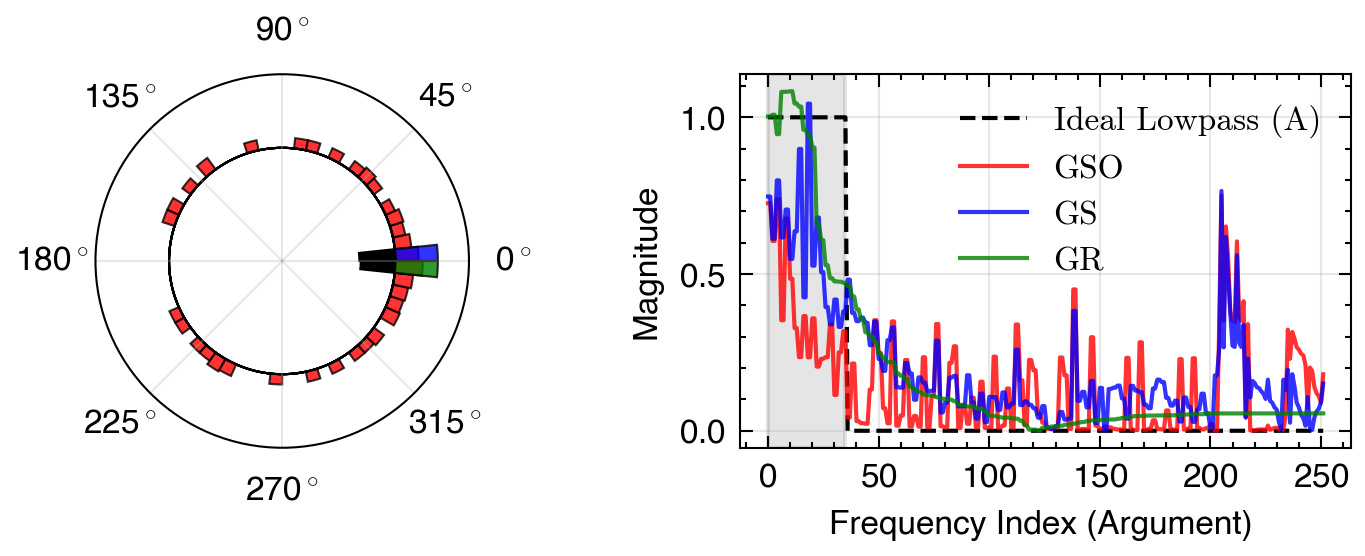}}
    \subfloat[Linear Phase Shift \label{fig:graph-phasefilter}]{%
       \includegraphics[width=0.5\linewidth]{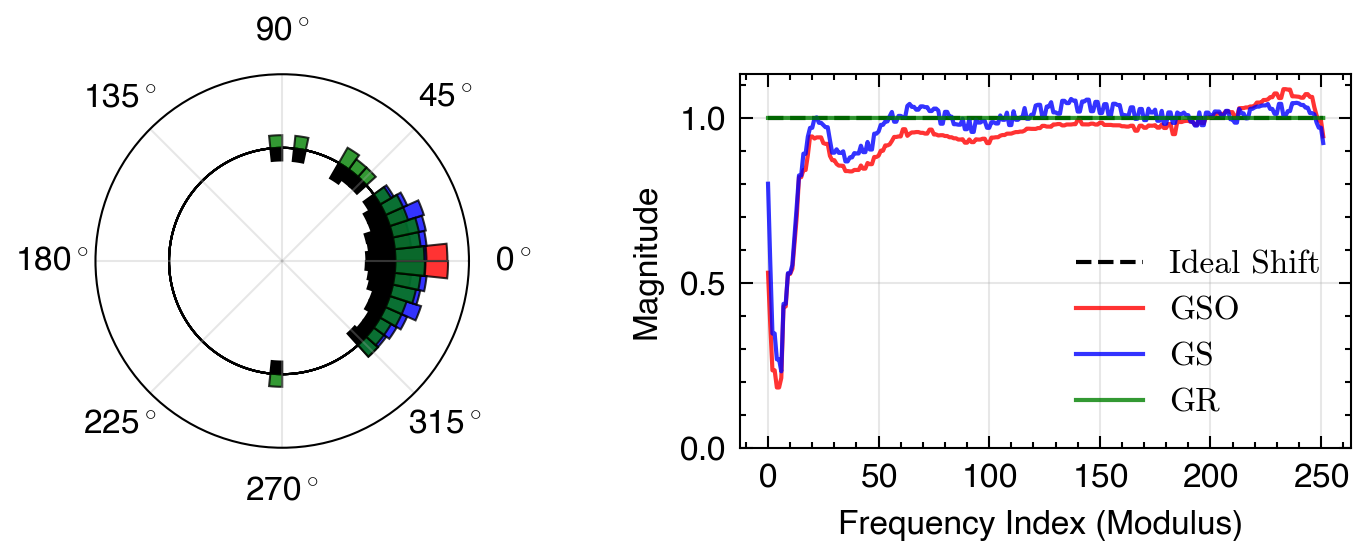}}
  \caption{(a) Wind field averaged over the period of year 2018.  (b) Temperature graph with directed edges colored according to orientation, and with undirected edges in black. (c) Spectral response (coefficients' angle and amplitude) of ideal modulus low pass, the graph sum (GS) filter, the graph rational (GR) filter and the GSO filter approximation of order 10. (d) Spectral response of ideal argument low pass, the GS filter, the GR filter and the GSO filter approximation of order 10. (e) Spectral response of ideal phase shift, the GS filter, the GR filter and the GSO filter approximation of order 10.} 
\end{figure*}

\subsection{Heat and Transport Kernels}

We illustrate the behavior of the graph heat and transport kernels when applied to spatially localized signals on the vortex graph. Specifically, we consider an initial graph signal (Fig.~\ref{fig:vortex-heat-transport-process}) defined as a Gaussian profile over the node coordinates:
\begin{equation*}
    \vc x_0[n] = e^{-\sigma \|\vc p[n] - \vc c\|^2},
\end{equation*}
where $\vc p[n]\in \mathbb{R}^2$ denotes the spatial coordinates of node $n$. We choose $\vc c$ the center of the Gaussian profile to be located within the lower part of the vortex, and $\sigma=0.1$ that controls the spatial spread of the signal.

Starting from $\vc x_0$, we can generate diffused signals at any time $t$ with the heat kernel $g_t^{\diamond}(\ma L^{\circ})$ as
\begin{eqnarray*}
    \vc x_t^{(1)} &=& g_t^{\diamond}(\ma L^{\circ})\vc x_0.
\end{eqnarray*}
Similarly, using the transport kernel $g^{\diamond}(\ma L^{\uparrow})$, we obtain
\begin{eqnarray*}
    \vc x_t^{(2)} &=& g_t^{\diamond}(\ma L^{\uparrow})\vc x_0.
\end{eqnarray*}
We generate signals for $t=100, 200, 300$ with $\alpha=0.005$ as shown in Fig.~\ref{fig:vortex-heat-transport-process}. 
As expected, the heat kernel progressively diffuses $\vc x_0$ isotropically over the graph, resulting in increasing smoothing. In contrast, the transport kernel propagates the signal along the underlying flow field while largely preserving its concentration.

We further examine the spectral behavior of these signals. Fig.~\ref{fig:vortex-heat-transport-process} shows the spectra of all the signals considered so far. As expected for diffusion phenomenon, the application of the heat kernel has a lowpass characteristic according to modulus ordering of frequencies. Conversely, the transport kernel primarily impacts the spectrum according to the argument ordering, reflecting phase modulation rather than attenuation. Indeed, we observe the nearly identical magnitude response across the signals filtered by the transport kernel.

\begin{figure*}[hbt!]
    \centering
    \captionsetup[subfigure]{justification=centering}
    \subfloat[Temperature regression performance\label{fig:temperature-regression-performance}]{%
       \includegraphics[width=0.46\linewidth]{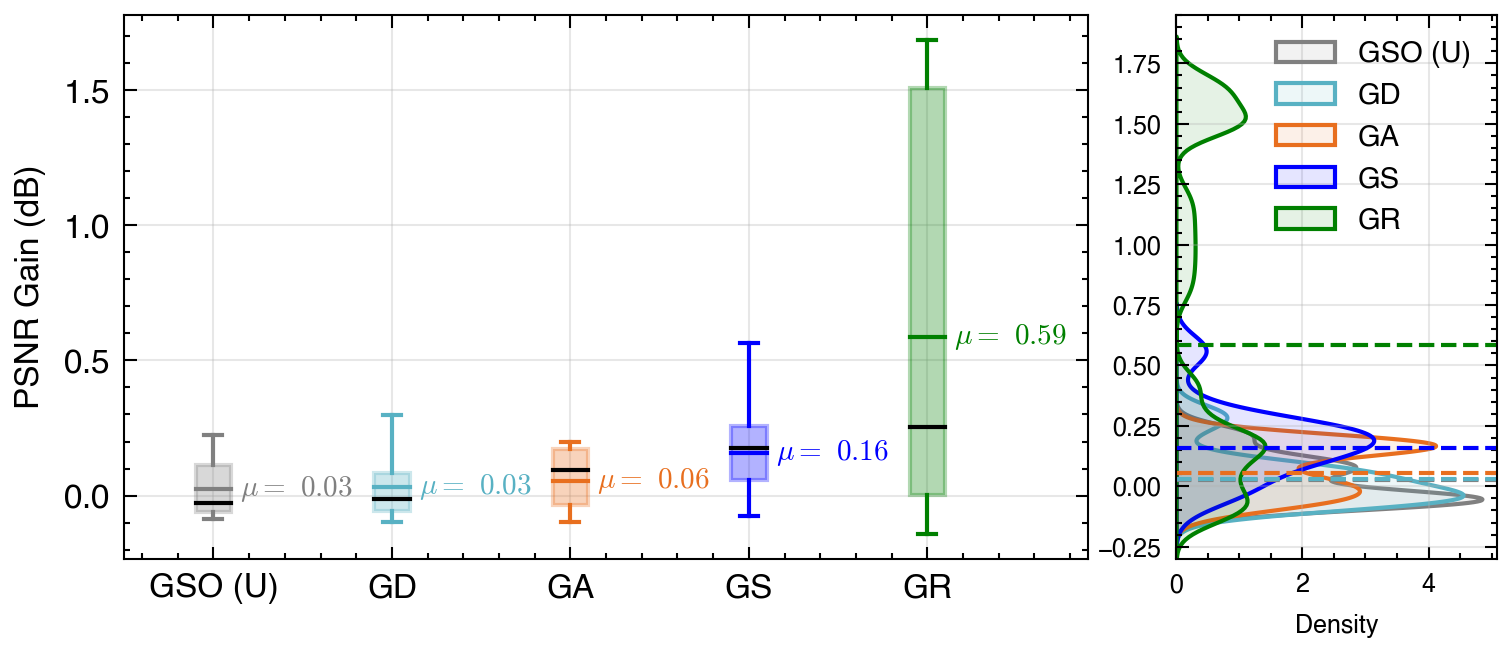}}
    \subfloat[Performance GS \label{fig:prediction-gso}]{
      \raisebox{13.6pt}{
        {%
      \includegraphics[width=0.239\linewidth]{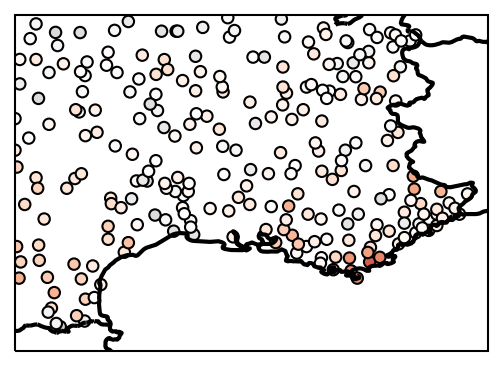}%
    }%
    }
    }
    \subfloat[Performance GR \label{fig:prediction-gr}]{
      \raisebox{13.6pt}{
      {%
      \includegraphics[width=0.300\linewidth]{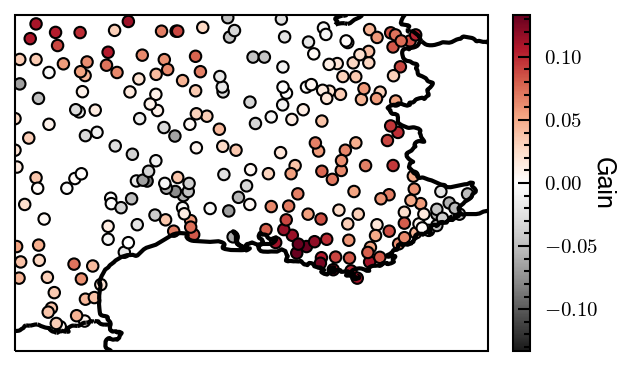}%
    }%
    }
    }
  \caption{(a) Performance gain of the undirected GSO, GD, GA, GS, and GR filters for temperature regression w.r.t directed GSO.  (b) Gain of GS-based temperature regression w.r.t directed GSO. (c) Gain of GR-based temperature regression w.r.t directed GSO.} 
\end{figure*}

\subsection{Case Study: Temperature Sensor Network}
\label{sec:graph-filters-approximation}
We construct a graph from a temperature sensor network that covers Côte d'Azur in the South-West of France~\cite{larvor2020meteonet}. The weather cast data is collected along the full year of 2018 and includes $N=252$ stations of interest. These stations constitute the nodes of the graph and we consider undirected edges according to the geographical mesh of the stations; i.e., the black edges Fig.~\ref{fig:distance-based-temperature}. The mesh is built via Delaunay triangulation, after which edges connecting distant stations (more than $1^\circ$ combined longitude-latitude distance) are removed. The undirected edge weights are set to 1. Then, leveraging wind directions information from the dataset (Fig.~\ref{fig:wind-field-temperature}), we build the set of directed edges. For each node, a directed edge to a neighboring node (according to the mesh) is added when the direction of the corresponding mesh edge is sufficiently aligned with the node's average wind direction; i.e., when their absolute angle difference is less than $\leq \pi/3$. These directed edges, shown in color in Fig.~\ref{fig:distance-based-temperature}, are set to unit weight. Finally, the adjacency matrix of the final graph is constructed as a convex combination of the undirected (mesh) and directed (wind) adjacency matrices. To guarantee diagonalizability, a small perturbation is applied to the resulting matrix, following \cite{seifert_digraph_2021}, which effectively adds two edges to the original set of 643 edges. In our experiments, the convex combination parameter is set to $0.6$, giving a slightly stronger weight to the wind direction. 
The resulting graph is shown in Fig.~\ref{fig:distance-based-temperature}, which is directed and weighted. We provide the first few eigenvectors according to our two proposed orderings in Supplementary.

We now want to demonstrate the feasibility of different graph filters. Specifically, we compare different filter designs using the sum filter (GS), the rational filter (GR), and the GSO-based polynomial filter. The approximation is performed by solving the least-squares problems described in Section~\ref{sec:filters}. The order of the polynomials is set to $K=10$ (results for order $K=3$ and $K=5$ can be found in supplementary). We evaluate the approximation quality by comparing the spectral response of the approximated filters with the ideal ones, both in terms of magnitude and phase response.
\subsubsection{Lowpass in the Diffusive Sense}
We use the ideal lowpass filter $\hat{\ma H}_\text{low}^\circ$ with $T$ such that \text{card}$(\mathcal{L}^\circ)=\floor{N/7}$. From Fig.~\ref{fig:graph-filterlowpass-D}, we see a better approximation from the GS compared to the GSO-based polynomial filter and the poorer approximating GR. However, the phase response is perturbed for the GSO filter, while GS and GR preserve ideal phase response, which are not introducing any phase shift.

\subsubsection{Lowpass in the Advective Sense}
We use the ideal lowpass filter $\hat{\ma H}_\text{low}^\uparrow$ with $T$ such that \text{card}$(\mathcal{L}^\uparrow)=\floor{N/7}$ and is an odd number. In Fig.~\ref{fig:graph-filterlowpass-A}, we observe that the GR filter yields the best approximation compared to GS and the GSO filter. Moreover, the phase response is perfectly preserved by GS and GR (no phase shift), while the GSO filter exhibits phase-shifting behavior.

\subsubsection{Linear Phase Shift}
We use the ideal linear phase shift filter $\hat{\ma H}_\text{shift}^{(q)}$ with $q=6$. This is chosen so that the smoothest spatial mode undergoes a phase shift of approximately half a period, corresponding to a shift across half of the map. Equivalently, for the smoothest (argument-wise) non-zero eigenvalue, $q\angle(\ma \Lambda[1])\approx \pi$. 
As shown in Fig.~\ref{fig:graph-phasefilter}, the GR approximation attains an almost perfect match to the ideal response both in magnitude and phase. The GS approximation preserves the phase response well while providing a more stable magnitude than a straightforward GSO polynomial approximation, which exhibits noticeable phase distortion and amplitude deviations. We further perform approximation of ideal filters for random phase filters, and graph Hilbert filter, as shown in Supplementary.

\subsection{Case Study: Temperature Regression}
\label{sec:temperature-regression}
Using now the temperature data of the Côte d'Azur weather stations over the period of year 2018 \cite{larvor2020meteonet}, we denote $\vc X$ the full data matrix of size $400 000 \times 252$, centered in time, and indexed as $\vc x_t$. In particular we select the timepoints for which the wind directions embedded in the graph are representative. Practically, we compute the circular correlation \cite{fisher1995statistical} across the stations, between the average wind directions (from which we constructed the graph) and the wind direction at each timepoint. The subset $S$ contains all timepoints with higher correlation than the $99.5$-th percentile leaving us with $\text{card}(S)=2000$ timepoints.
We aim to characterize the temporal process undergone by $\vc x_t$ through the graph in Fig.~\ref{fig:distance-based-temperature}. To this end, we consider pairs of subsequent time steps of temperature measurements, $\vc x_t$ and $\vc x_{t+d}$ with $t\in S$ and provide an estimate of $\vc x_{t+d}$ from $\vc x_t$ through a chosen graph filter $h$:
\begin{equation*}
    {\vc y}_{t+d} = h(\ma L) \vc x_t.
\end{equation*}

We consider the sum (GS), rational (GR), heat (GD), transport (GA), directed graph shift operator (GSO), and undirected graph shift operator (GSO U) filters.
To determine the coefficients of each filter, we consider the average normalized mean square error (NMSE) over all time steps:
\begin{equation*}
    \text{NMSE} := \frac{1}{\text{card}(S)}\sum_{t\in S} \frac{||\vc y_{t+d} - \vc x_{t+d}||_2^2}{||\vc x_{t+d}||_2^2}.
\end{equation*}
In this experiment, we set $d = 40$ timepoints (equivalent to 20 min), allowing the temperature to sufficiently vary. If $d$ is small, the obtained graph filter $h$ amounts to the identity filter. Additionally, the order of all filters, except for heat and transport kernels that are order $1$, are set to $4$.

In Fig.~\ref{fig:temperature-regression-performance}, we show the prediction performance of the different filters in terms of Peak Signal to Noise Ratio (PSNR) gain with respect to the directed GSO. Firstly the results show that the rational graph filter GR outperforms all other graph filters, demonstrating its effectiveness in capturing the underlying temperature evolution through an average PSNR gain of 0.59 --- as a reference the average PSNR of the directed GSO is 17. Notably, the transport kernel and sum filter also allow an average PSNR gain of 0.06 and 0.16 respectively, compared to the directed GSO, suggesting that the advection process controlled by both graph filters play a significant role in the temperature dynamics. Surprisingly, the undirected GSO also seem to marginally improve prediction, indicating that the integration of directional information through the directed GSO may actually be detrimental and may hint at the directionality information not being properly exploited by the directed GSO. Lastly, we show a sample of the performance gain with respect to the GSO filter for the sum filter and the rational graph filter in Figs.~\ref{fig:prediction-gso} and~\ref{fig:prediction-gr}, respectively. For the sum filter, we observe mainly improvements on the south coast, while for the rational filter, we see a more spread and stronger improvements.

\subsection{Implementation}
Documented code is available on this Github Repository 
\footnote{\url{https://github.com/MIPLabCH/GraphDiffusionAdvection}}. The implementation is provided in Python and includes all algorithms and figures illustrated in this paper. 

\section{Discussion \& Outlook}
In this work, we took inspiration from the diffusion-advection operator from physics and its continuous-domain eigendecomposition to reconsider the directed graph Laplacian as an operator that integrates these two concepts. We derived how this naturally introduces new orderings of the complex-valued eigenvalues based on their modulus and argument. Subsequently, this extends to the specification of different smoothness measures; i.e., either over the undirected part of the graph, or along the directed part. This differs from previous work \cite{sandryhaila_discrete_2014, sardellitti_graph_2017, shafipour_directed_2019}, where the proposed smoothness measures do not distinguish between these two possibilities. 
Furthermore, we introduce graph filters corresponding to heat and transport kernels, and also two new classes of graph filters, the sum and rational graph filters. As for conventional GSO-based graph filters, they  can be optimized to approximate ideal filters for a given order. We show through experiments that the proposed graph filters surpass simple GSO polynomial in approximating ideal filters. Finally, we apply the sum and rational graph filters to a temperature regression task on a wind-derived graph and demonstrate that they outperform simple GSO polynomial filter, with the rational graph filter providing the best regression results. In this experiment, we highlight that directionality necessary to capture signal dynamics is not fully exploited by the GSO filter, even if it is defined on the same directed graph, and instead requires the flexibility offered by the newly introduced diffusion and advection operators.


Our contribution also lies in connecting directed GSP with a partial differential equation, namely the diffusion-advection equation,in the same way that undirected GSP is guided by its link with the diffusion equation. This establishes a key conceptual distinction from approaches that symmetrize \cite{brefeld_graph_2020, wei_hermitian_2024} and thus implicitly consider diffusion only. Moreover, in several instances -- notably in filter design -- we show that the complex-valued eigendecomposition is not a technical obstacle, but rather a meaningful and exploitable feature. This differentiates our work from the ones that impose real-valued bases \cite{sardellitti_graph_2017, shafipour_directed_2019, chen_graph_2023}. 

Our work can also be related to the polar decomposition. Indeed, similarities may be drawn with \cite{ji2026digraph} where the eigenvalues of the two factorized matrices lead to complex exponential and real-valued eigenvalues interpreted as rotation and scaling, respectively. In contrast, our diffusion ($\ma L^{\circ}$) and advection operators ($\ma L^{\uparrow}$) draw interpretation from a physical process and recombine by addition into the original GSO, rather than by multiplication. We also documented a set of desirable properties useful for the construction of graph filters. Rather than positioning these perspectives in opposition, our framework can provide complementary physical insight into existing GSO decomposition in directed GSP, such as the SVD-based GFT \cite{chen_graph_2023}, and the polar decomposition-based GFT \cite{kwak_frequency_2024, ji2026digraph} as both rely on the eigendecomposition of $\ma L\ma L^T$ and $\ma L^T\ma L$, which in turn can be seen as a consecutive application of the diffusion-advection operator in one direction and its opposite. 


Importantly, we do not discretize the continuous diffusion-advection operator, as it is known that the operator does not admit an immediate extension to graphs. Existing approaches that substitute differential operators by graph counterparts \cite{chamberlain2021beltrami} remain limited by concerns of satisfying the Leibniz rule and a formal definition of vector fields on graphs. Instead, we derive $\ma L^\circ$ and $\ma L^\uparrow$ from the spectral structure of the continuous diffusion-advection operator, seeking to preserve the effective dynamical properties associated with diffusion and advection rather than discretizing the differential operators themselves. In the continuous setting, the real part of the spectrum governs modal dissipation through exponential decay, while the imaginary part governs modal phase evolution associated with transport. Correspondingly, \(\ma L^\circ\) isolates the spectral contribution responsible for attenuation and smoothing effects, whereas \(\ma L^\uparrow\) isolates the contribution responsible for directional transport and oscillatory behavior. We further strengthen the analogy through the case where $\ma L$ is a normal operator, for which $\ma L^\circ$ is symmetric and $\ma L^\uparrow$ is skew-symmetric, recovering established definitions of graph Laplacian \cite{ortega_graph_2018} and graph advection \cite{chapman2015advection}, respectively. We emphasize, however, that these constructions provide either a graph diffusion or graph advection directly from the graph adjacency matrix, whereas our framework extracts both components from the directed graph Laplacian. Furthermore, we show that $\ma L^{\uparrow}$ satisfy divergence-free property, consistent, with the interpretation of advection as a transport induced by an incompressible velocity field. Beyond this theoretical consistency, we provide empirical evidence demonstrating that $\ma L^\circ$ captures symmetric (diffusive) behavior, whereas $\ma L^\uparrow$ captures directional (advective) behavior. Taken together, these results support the proposed decomposition as a meaningful and structurally grounded analogue of the diffusion-advection dynamics on graphs.

Distinctively from the GSO-based filter, the sum and rational filters require a spectral decomposition of the underlying graph Laplacian, which can be computationally expensive for large graphs. However, in practice, these filters can be approximated using Lanczos solver \cite{lanczos1950iteration} that compute only a subset of the eigenvalues and eigenvectors, particularly those corresponding to the smallest real parts and largest imaginary parts of the spectrum, which are most relevant for capturing the graph diffusion and advection operators.


The new insights provided by this work can spur further advances in directed GSP. First, the ordering and smoothness definitions can be readily used in regularization terms for denoising \cite{chen_signal_2015, pang_graph_2017}, learning of graphs structure \cite{dong_learning_2016, kalofolias2016learn} and signal separation \cite{mohammadi_graph_2023, yarandi_new_2024}. 
Second, the presented graph filters could serve as building blocks for GNNs~\cite{bronstein2017geometric}. While works on GNNs applying the concept of diffusion-advection exist, they are not taking the same perspective as in this work. Indeed, works from \cite{eliasof2024feature, wu2025supercharging} employ the term advection either to refer to learnable directionality parameters or to message passing mechanism using the adjacency matrix. Other work \cite{berlureau2026advection} suggests to change dynamics of the diffusion by biasing direction through nodal potential, this nodal potential is learned and applied to the symmetric graph Laplacian, preserving symmetry of the GSO. Thus, the dynamics described by the modulated operator more closely resemble that of anisotropic diffusion, or directionality preferential, diffusion. Importantly, the proposed graph filters in our work are shown to translate a variety of operations in the graph spectral domain, such as argument and modulus low, high pass filtering, phase shifting to the nodal domain. These operations, could be leveraged to improve features propagation on the graph, taking a different stance in tackling oversquashing and oversmoothing \cite{qureshi2023limits}, which are important concerns in the field of GNNs.



\appendices

\section{Properties of the New Operators \texorpdfstring{$\ma L^{\circ}$ and $\ma L^{\uparrow}$}{L-circ and L-up}}
\label{app:prop}

\subsection{Proof of Proposition~\ref{prop:real}}

\begin{proof}
We recall that the columns $\ma U^{-1}$, denoted $\vc v_n$ are also left eigenvectors of $\ma L$, therefore $\vc v_n^T \ma L  = \ma \Lambda[n] \vc v_n^T$, and as such there exist a conjugate index $n^*$ such that $\vc v_{n^*} = \vc v_n^*$ and notably we have $\vc u_{n^*} = \vc u_n^*$ for the same pair $(n,n^*)$. It is then straightforward to see that the spectral projectors $\vc u_n \vc v_n^T$ also admits a conjugate pair at index $n^*$. By viewing $\ma L^{\circ}$ and $\ma L^{\uparrow}$ using its spectral projectors, we have
\begin{equation*}
    \ma L^{\circ} = \sum_{n} \ma \Lambda_R[n] \vc u_n \vc v_n^T, \quad \ma L^{\uparrow} = \sum_{n} j\ma \Lambda_I[n] \vc u_n \vc v_n^T.
\end{equation*}
Denoting $\mathcal{R}$ the indices for real eigenvalues and $\mathcal{I}$ the indices for complex eigenvalues, we can rewrite the above expressions as
\begin{align*}
    \ma L^{\circ} &= \sum_{n\in\mathcal{R}} \ma \Lambda_R[n] \vc u_n \vc v_n^T + \frac{1}{2}\sum_{n\in\mathcal{I}} \ma \Lambda_R[n] \vc u_n \vc v_n^T + \ma \Lambda_R[n^*] \vc u_{n^*} \vc v_{n^*}^T\\
    &=\sum_{n\in\mathcal{R}} \ma \Lambda_R[n] \vc u_n \vc v_n^T + \sum_{n\in\mathcal{I}} \ma \Lambda_R[n] \Re\left(\vc u_n \vc v_n^T\right) \in \mathbb{R}^{N\times N},
\end{align*}
and 
\begin{align*}
    \ma L^{\uparrow} &= \sum_{n\in\mathcal{I}} j\ma \Lambda_I[n] \vc u_n \vc v_n^T + j\ma \Lambda_I[n^*] \vc u_{n^*} \vc v_{n^*}^T \\
    &= \frac{1}{2}\sum_{n\in\mathcal{I}} j\ma \Lambda_I[n] (\vc u_n \vc v_n^T - \vc u_{n^*} \vc v_{n^*}^T) \\
    &= -\sum_{n\in\mathcal{I}} \ma \Lambda_I[n] \Im\left(\vc u_n \vc v_n^T\right) \in \mathbb{R}^{N\times N},
\end{align*}
which concludes the proof.
\end{proof}

\subsection{Proof of Proposition \ref{prop:zero-row-sum}}
\begin{proof}
This is a direct consequence of the construction of $\ma L^{\circ}$ and $\ma L^{\uparrow}$ since both have the same eigenvectors as $\ma L$ and, in particular, the eigenvalue $0$ associated to the constant eigenvector $\vc 1$.
\end{proof}

\subsection{Proof of Proposition \ref{prop:eigenvalues}}
\begin{proof}
The eigenvalues of the directed Laplacian $\ma \Lambda[n]$ have non-negative real parts; i.e.,  a graph with a single connected component only has one zero eigenvalue and the remaining eigenvalues are positive \cite{singh_graph_2016}. Therefore, the eigenvalues of $\ma L^{\circ}$ are real and also non-negative, leading to a rank of $N-1$. On the other hand, the eigenvalues of $\ma L^{\uparrow}$ are purely imaginary by construction, which also implies a rank of at most $N-1$.
\end{proof}

\subsection{Proof of Proposition \ref{prop:trace}}
\begin{proof}
The trace of $\ma L^{\circ}$ and $\ma L^{\uparrow}$ can be computed using their eigenvalues as
\begin{equation*}
    \Tr(\ma L^{\circ}) = \sum_{n=1}^N \ma \Lambda_R[n], \quad \Tr(\ma L^{\uparrow}) = \sum_{n=1}^N j\ma \Lambda_I[n].
\end{equation*}
Recall that the eigenvalues of $\ma L$ have non-negative real parts, therefore $\Tr(\ma L^{\circ}) > 0$. Additionally, since the eigenvalues are either real or arranged in conjugate pairs, we have that $\sum_{n=1}^N \ma \Lambda_I[n] = 0$, leading to $\Tr(\ma L^{\uparrow})=0$.
\end{proof}

\subsection{Proof of Corollary \ref{corr:maximal-poly-diff}}
\begin{proof}
Denoting $\mathcal{R}$ the set of indices for the real eigenvalues and $\mathcal{I}$ the set of pairs of indices for conjugate pairs. The characteristic polynomial of $\ma L^{\circ}$ is given by
\begin{align*}
    p_{\ma L^{\circ}}(\lambda) &= \prod_{n\in\mathcal{R}} (\lambda - \ma \Lambda_R[n]) \prod_{(n_1,n_2)\in\mathcal{I}} (\lambda - \ma \Lambda_R[n_1])(\lambda - \ma \Lambda_R[n_2])\\
    &= \prod_{n\in\mathcal{R}} (\lambda - \ma \Lambda_R[n]) \prod_{(n_1,n_2)\in\mathcal{I}} (\lambda - \ma \Lambda_R[n_1])^2,
\end{align*}
as such the minimal polynomial 
\begin{equation*}
    m_{\ma L^{\circ}}(\lambda) = \prod_{n\in\mathcal{R}} (\lambda - \ma \Lambda_R[n]) \prod_{(n_1,n_2)\in\mathcal{I}} (\lambda - \ma \Lambda_R[n_1]),
\end{equation*}
and has degree $N-\text{card}(\mathcal{I})=N-N_C$. Therefore by Cayley-Hamilton theorem \cite{horn2012matrix} we have that $m_{\ma L^{\circ}}(\ma L^{\circ})=0$ which concludes our claim.
\end{proof}

\subsection{Proof of Corollary \ref{corr:maximal-poly-adv}}
\begin{proof}
The proof follows the same steps as the one of Corollary~\ref{corr:maximal-poly-diff} by considering the characteristic polynomial of $\ma L^{\uparrow}$ given by
\begin{align*}
    p_{\ma L^{\uparrow}}(\lambda) &= \prod_{n\in\mathcal{R}} (\lambda - \ma \Lambda_I[n]) \prod_{(n_1,n_2)\in\mathcal{I}} (\lambda - \ma \Lambda_I[n_1])(\lambda - \ma \Lambda_I[n_2])\\
    &= \lambda^{|\mathcal{R}|}\!\prod_{(n_1, n_2)\in\mathcal{I}} (\lambda - \ma \Lambda_I[n_1])(\lambda + \ma \Lambda_I[n_1]).
\end{align*}
In turn, the corresponding minimal polynomial reads as
\begin{equation*}
    m_{\ma L^{\uparrow}}(\lambda) = \lambda \prod_{(n_1, n_2)\in\mathcal{I}} (\lambda - \ma \Lambda_I[n_1])(\lambda + \ma \Lambda_I[n_1]),
\end{equation*}
with degree $N-\text{card}(\mathcal{R}) + 1=N-N_Z + 1$. Therefore by Cayley-Hamilton theorem \cite{horn2012matrix} we have that $m_{\ma L^{\uparrow}}(\ma L^{\uparrow})=0$ which concludes our claim.
\end{proof}

\section{Properties of \texorpdfstring{$\ma L^{\circ}$ and $\ma L^{\uparrow}$\\ for Normal Graph Operator}{L-circ and L-up for Normal Graph Operator}}
\label{app:prop-normal}

\subsection{Proof of Proposition \ref{prop:normal-symmetry-asymmetry}}
\begin{proof}
Let the directed graph Laplacian $\ma L$ be a normal operator, hence $\ma L$ is unitarily diagonalizable and admits the following eigendecomposition $\ma L=\ma U\ma \Lambda\ma U^H$\cite{axler2024linear}. Consequently, we have, 
\begin{eqnarray*}
(\ma L^\circ)^H &=& \ma U\ma \Lambda\ma U^H=\ma L^{\circ}\\    
(\ma L^\uparrow)^H &=& \ma U(\ma \Lambda)^H\ma U^H=-\ma L^{\uparrow}.
\end{eqnarray*}
Additionally, 
\begin{eqnarray*}
    \frac{\ma L + \ma L^H}{2} &=& \ma U\left(\frac{\ma \Lambda + (\ma \Lambda)^H}{2}\right)\ma U^H = \ma U\Re(\ma \Lambda)\ma U^H=\ma L^\circ \\
    \frac{\ma L - \ma L^H}{2} &=& \ma U\left(\frac{\ma \Lambda - (\ma \Lambda)^H}{2}\right)\ma U^H = \ma U\Im(\ma \Lambda)\ma U^H=\ma L^\uparrow,
\end{eqnarray*}
which concludes the proof.
\end{proof}

\subsection{Proof of Proposition \ref{prop:normal-divergence-free}}
\begin{proof}
The proof immediately follows from \cite{wu2005algebraic}, stating that if the corresponding directed graph Laplacian is normal (i.e., $\ma L \ma L^H = \ma L^H \ma L$), then the graph is Eulerian. Consequently, the graph corresponding to $(\ma L - \ma L^T)/2$ is also Eulerian. From the expression for divergence, we derive that $\nabla\cdot \ma L^{\uparrow} = 0$, concluding the proof.
\end{proof}


\section{Smoothness and Ordering of Eigenvectors}
\label{app:smoothness}

\subsection{Proof of Proposition \ref{prop:smoothness-ordering}}
\begin{proof}
The proof is straightforward as we directly evaluate the smoothness of each eigenvectors as follows
\begin{align*}
    \mathcal{S}^{\uparrow}(\vc u_n) &= ||\ma L^{\uparrow} (\ma L^{\circ})^{\dagger}\vc u_n||_2^2 \\
    &= \left(\ma L^{\uparrow} (\ma L^{\circ})^{\dagger}\vc u_n\right)^H\left(\ma L^{\uparrow} (\ma L^{\circ})^{\dagger}\vc u_n\right)\\
    &= -j\vct \delta_n^H\ma \Lambda_R^{\dagger}\ma \Lambda_I\ma U^H\ma Uj\ma \Lambda_I \ma \Lambda_R^{\dagger}\vct \delta_n \\
    &=|\ma \Lambda_I[n]|^2|\ma \Lambda_R^{\dagger}[n]|^2 \vc u_n^H\vc u_n\\
    &=|\ma \Lambda_I[n]|^2|\ma \Lambda_R^{\dagger}[n]|^2, \\[6pt]
    \mathcal{S}^{\circ}(\vc u_n) &= ||\ma L\vc u_n||_2^2\\
    &= (\ma L\vc u_n)^H(\ma L\vc u_n)\\
    &= \vct \delta_n^H \ma\Lambda^H\ma U^H\ma U\ma \Lambda\vct \delta_n \\
    &=|\ma \Lambda[n]|^2.
\end{align*}
We retrieve both $|\tan(\angle\ma \Lambda[n])|^2$ and $|\ma \Lambda[n]|^2$. Since the square function is increasing on $[0,\infty)$, $\mathcal{S}^\circ$ and $\mathcal{S}^\uparrow$ yield the modulus and argument orderings, respectively. 
\end{proof}

\section{Properties of the Graph Filters}
\label{app:graph-filters}

\subsection{Proof of Proposition \ref{prop:LSI-filters}}
\begin{proof}
    Given that $\ma L$, $\ma L^{\uparrow}$, and $\ma L^{\circ}$ are jointly diagonalizable, it follows that any polynomial or rational function of $\ma L$, $\ma L^{\uparrow}$, or $\ma L^{\circ}$ are also jointly diagonalizable with $\ma L$, $\ma L^\circ$, and $\ma L^\uparrow$. In other words, the filters commute with these operators; e.g., for $\ma L$ we have: 
    \begin{eqnarray*}
        \ma L \, \tilde{h}_{\text{sum}}(\ma L^{\circ}, \ma L^{\uparrow}) &=& \tilde{h}_{\text{sum}}(\ma L^{\circ}, \ma L^{\uparrow}) \, \ma L, \\
        \ma L \, \tilde{h}_{\text{rat}}(\ma L^{\uparrowcirc}) &=&\tilde{h}_{\text{rat}}(\ma L^{\uparrowcirc}) \, \ma L,
    \end{eqnarray*}
showing that both $\tilde{h}_{\text{sum}}$ and $\tilde{h}_{\text{rat}}$ are LSI filters.
\end{proof}

\subsection{Proof of Proposition \ref{prop:poly-real-coefs}}

\begin{proof}
Consider desired ideal filter $\hat{\ma H}_{\text{spec}}\in\mathbb{C}^{N\times N}$ to be a spectral filter such that for every graph signal $\vc x\in\mathbb{R}^N$, the filtered signal~\upshape IGFT$\{\hat{\ma H}_{\text{spec}}\hat{\vc x}\}\in\mathbb{R}^N$. Let us also denote $\mathcal{R}$ the set of indices for the real eigenvalues and $\mathcal{I}$ the set of pairs of indices for conjugate pairs. 

Starting with the coefficients $\vct \eta$ for $\tilde{h}_{\text{rat}}(\ma L^{\uparrowcirc})$,
we simplify notations by setting $\ma V=\ma V_K(j{\ma \Lambda_I \ma \Lambda_R^{\dagger}})^{\dagger}$. We derive
\begin{align}
    \vct \eta[k] &= \sum_{m=1}^{N}\ma V[k, m] \hat{\ma H}_{\text{spec}}[m, m] \nonumber \\
    &= \sum_{m,n\in \mathcal{I}} \ma V[k, m] \hat{\ma H}_{\text{spec}}[m, m] + \ma V[k, n] \hat{\ma H}_{\text{spec}}[n, n] \nonumber \\
    &+ \sum_{m\in \mathcal{R}} \ma V[k, m] \hat{\ma H}_{\text{spec}}[m, m] \nonumber \\ 
    &= \sum_{m,n\in \mathcal{I}} \ma V[k, m] \hat{\ma H}_{\text{spec}}[m, m] + (\ma V[k, m] \hat{\ma H}_{\text{spec}}[m, m])^* \nonumber \\
    &+ \sum_{m\in \mathcal{R}} \ma V[k, m] \hat{\ma H}_{\text{spec}}[m, m] \nonumber \\ 
    &= \sum_{m, n\in \mathcal{I}} 2\Re(\ma V[k, m] \hat{\ma H}_{\text{spec}}[m, m]) \nonumber \\
    &+ \sum_{m\in \mathcal{R}} \ma V[k, m] \hat{\ma H}_{\text{spec}}[m, m] \nonumber,
\end{align}
and thus $\vct \eta\in\mathbb{R}^K$.

Similarly, by replacing $\ma V$ by $\ma C_K(\ma \Lambda_R, j\ma \Lambda_I)^{\dagger}$ and considering the same ideal filter $\hat{\ma H}_{\text{spec}}$, one can show that the coefficients 
$\vct \alpha, \vct \beta\in\mathbb{R}^K$.

\end{proof}

\section*{Supplementary Materials}
Additional plots are listed in supplementary materials.


\ifCLASSOPTIONcaptionsoff
  \newpage
\fi



\bibliographystyle{ieeetr}
%





\bibliography{references}

\end{document}

%% file: macros.tex
\newcommand{\vc}[1]{{\mathbf #1}}
\newcommand{\vct}[1]{\pmb{#1}}
\newcommand{\ma}[1]{{\mathbf #1}}

\newcommand{\Tr}{\mathrm{Tr}}
\newcommand{\uparrowcirc}{
\tikz[baseline=0ex]{
\draw[->] (0,-0.08) -- (0,0.2);
\draw (0,0.05) circle (0.06);
}}

\DeclareMathOperator*{\argmin}{arg\,min}

\DeclarePairedDelimiter\floor{\lfloor}{\rfloor}

\newlength{\subcolumnwidth}

\newcommand{\nextsubcolumn}[1][]{%
  \cr\noalign{\hfill}
  \if\relax\detokenize{#1}\relax\else\hsize=#1\setlength{\subcolumnwidth}{\hsize}\fi
}